\documentclass[11pt, letterpaper]{article}

\def\final{1}

\usepackage{times}
\usepackage{paralist}
\usepackage{amsfonts}
\usepackage{amsmath}
\usepackage{amsthm}
\usepackage{latexsym}
\usepackage{url,hyperref}
\usepackage{array}
\usepackage{amssymb}
\usepackage{accents}
\usepackage{xcolor}
\usepackage[font={small}]{caption}
\usepackage[algoruled]{algorithm2e}

\usepackage[normalem]{ulem}

\usepackage{float}
\usepackage{graphicx,subfigure}

\usepackage{verbatim}

\ifdefined\final
\usepackage{microtype}
\fi

\graphicspath{{graphics/}}

\usepackage[margin=1in]{geometry}

\definecolor{darkblue}{rgb}{0,0,0.38}
\definecolor{darkred}{rgb}{0.6,0,0}
\definecolor{darkgreen}{rgb}{0.1,0.35,0}
\hypersetup{colorlinks, linkcolor=darkblue, citecolor=darkgreen, urlcolor=darkblue}

\newtheorem{theorem}{Theorem}[section]
\newtheorem*{theorem*}{Theorem}

\newtheorem*{conjecture*}{Conjecture}

\newtheorem*{remark*}{Remark}
\newtheorem{lemma}[theorem]{Lemma}
\newtheorem*{lemma*}{Lemma}
\newtheorem{proposition}[theorem]{Proposition}
\newtheorem*{proposition*}{Proposition}
\newtheorem{claim}[theorem]{Claim}
\newtheorem*{claim*}{Claim}

\newtheorem*{corollary*}{Corollary}

\newtheorem*{property*}{Property}
\newenvironment{proofof}[1]{\begin{proof}[#1]}{\end{proof}}
\theoremstyle{definition}

\DeclareMathOperator{\cost}{cost}
\newcommand{\cc}[1]{c(#1)}

\ifdefined\final
\newcommand{\rico}[1]{}
\newcommand{\neil}[1]{}
\newcommand{\todo}[1]{}

\else
\newcommand{\rico}[1]{\emph{\color{red}rico: #1}}
\newcommand{\neil}[1]{\emph{\color{darkgreen}{Neil: #1}}}
\newcommand{\todo}[1]{{\color{blue}\emph{(#1)}}}

\fi

\def\b1{{\bf 1}}

\newcommand{\grphx}{\mathcal{X}}
\newcommand{\grphy}{\mathcal{Y}}

\newcommand{\compof}[1]{\Gamma(#1)}

\newcommand{\comp}{\mathcal{K}}
\newcommand{\compx}{\compof{\grphx}}
\newcommand{\compF}{\compof{\grphx - F}}

\newcommand{\slk}{h}
\newcommand{\slkx}{h_{\grphx}}
\newcommand{\slkF}{h_{\bar{F}}}
\newcommand{\slkB}{h_{\bar{B}}}
\newcommand{\slkles}[1]{h_{\scriptsize\xoverline[0.4]{#1}}}
\newcommand{\slkless}[1]{h_{\scriptsize\xoverline{#1}}}

\newcommand{\dirE}{\vec{E}}

\newcommand{\remove}{-}
\newcommand{\contract}{\circledast}

\newcommand{\naturals}{\mathbb{N}}

\newcommand{\minimizers}[1]{\mathcal{S}_{#1}}

\newcommand{\grd}{E(\grphx)}
\newcommand{\core}{K}

\newcommand{\remP}{B_{\text{\it rem}}}
\newcommand{\remPcore}{B_{\text{\it rem}}^\core}

\newcommand{\gammoid}{\mathfrak{G}}
\newcommand{\bgammoid}{\gammoid'}
\newcommand{\brank}{\rho'}
\newcommand{\rank}{\rho}

\newcommand{\R}{\mathbb{R}}

\newcommand{\harmonic}[1]{H({#1})}

\newcommand{\pot}[2]{\Phi_{#1}(#2)}
\newcommand{\Pot}{\pot{\core}{\grphx}}
\newcommand{\E}[1]{\mathbb{E}\{#1\}}

\newcommand{\witness}[3]{W(#3)}
\newcommand{\Witness}[1]{\witness{\core}{\grphx}{#1}}
\newcommand{\primary}{core}
\newcommand{\secondary}{cleanup}

\newcommand{\card}[1]{\lvert#1\rvert}

\newcommand{\bases}[1]{\mathcal{B}'_{#1}}
\newcommand{\minrems}[1]{\mathcal{B}_{#1}}

\renewcommand{\Pr}[1]{\mathbb{P}\{#1\}}

\newcommand{\afront}{a^{\text{\rm f}}}
\newcommand{\aback}{a^{\text{\rm b}}}

\DeclareMathOperator{\supp}{supp}

\makeatletter
\newsavebox\myboxA
\newsavebox\myboxB
\newlength\mylenA

\newcommand*\xoverline[2][0.75]{%
    \sbox{\myboxA}{$\m@th#2$}%
    \setbox\myboxB\null%
    \ht\myboxB=\ht\myboxA%
    \dp\myboxB=\dp\myboxA%
    \wd\myboxB=#1\wd\myboxA%
    \sbox\myboxB{$\m@th\overline{\copy\myboxB}$}%
    \setlength\mylenA{\the\wd\myboxA}%
    \addtolength\mylenA{-\the\wd\myboxB}%
    \ifdim\wd\myboxB<\wd\myboxA%
       \rlap{\hskip 0.5\mylenA\usebox\myboxB}{\usebox\myboxA}%
    \else
        \hskip -0.5\mylenA\rlap{\usebox\myboxA}{\hskip 0.5\mylenA\usebox\myboxB}%
    \fi}
\makeatother

\title{Matroids and Integrality Gaps \\
for Hypergraphic Steiner Tree Relaxations}

\date{M.I.T. \\[1.5em] \today}

\author{%
Michel X.~Goemans\thanks{E-mail: {\tt goemans@math.mit.edu}. Supported by NSF grants CCF-1115849 and CCF-0829878, and by
ONR grant N00014-11-1-0053.}
\and Neil Olver\thanks{E-mail: {\tt olver@math.mit.edu}. Supported by NSF grant CCF-1115849.}
\and Thomas Rothvo\ss\thanks{E-mail: {\tt rothvoss@math.mit.edu}. Supported by the Alexander von Humboldt Foundation within the Feodor Lynen program, by ONR grant N00014-11-1-0053 and by NSF contract
CCF-0829878.}
\and Rico Zenklusen\thanks{E-mail: {\tt ricoz@math.mit.edu}.
Supported by NSF grants CCF-1115849 and CCF-0829878, and by
ONR grants N00014-11-1-0053 and N00014-09-1-0326.
}}

\begin{document}

\begin{titlepage}
\maketitle
\thispagestyle{empty}
\begin{abstract}

Until recently, LP relaxations have only played a very limited role in
the design of approximation algorithms for the Steiner tree problem.
In particular, no (efficiently solvable) Steiner tree
relaxation was known to have an integrality gap bounded away from $2$,
before Byrka et al.~\cite{byrka_2011_steiner} showed an upper bound
of $\approx 1.55$ of a hypergraphic LP relaxation 
and presented a $\ln(4)+\epsilon\approx 1.39$ approximation based
on this relaxation.
Interestingly, even though their approach is LP based, they do not
compare the solution produced against the LP value. 

We take a fresh look at hypergraphic LP relaxations
for the Steiner tree problem---one that heavily exploits methods
and results from the theory of matroids and submodular
functions---which leads 
to stronger integrality gaps, faster algorithms, and a variety of
structural insights of independent interest.
More precisely, along the lines of the algorithm of Byrka et
al.~\cite{byrka_2011_steiner}, we present a deterministic
$\ln(4)+\epsilon$ approximation that compares against the LP value and
therefore proves a matching $\ln(4)$ upper bound on the
integrality gap of hypergraphic relaxations.

Similarly to~\cite{byrka_2011_steiner}, we
iteratively fix one component and update the LP solution.
However, whereas in~\cite{byrka_2011_steiner} the 
LP is solved at every iteration after contracting a
component, we show how feasibility can be maintained by
a greedy procedure on a well-chosen matroid.
Apart from avoiding the expensive step of solving a hypergraphic
LP at each iteration, our algorithm can be analyzed
using a simple potential function.
This potential function gives an easy means
to determine stronger approximation guarantees and integrality
gaps when considering
restricted graph topologies.
In particular, this readily leads to a $\frac{73}{60}\approx 1.217$
upper bound on the integrality gap of hypergraphic relaxations
for quasi-bipartite graphs.

Additionally, for the case of quasi-bipartite graphs,
we present a simple algorithm to transform
an optimal solution to the bidirected cut relaxation
to an optimal solution of the hypergraphic relaxation,
leading to a fast $\frac{73}{60}$ approximation
for quasi-bipartite graphs.
Furthermore, we show how the separation problem of the hypergraphic
relaxation can be solved by computing maximum flows, which provides a
way to obtain a fast independence oracle for the matroids that we use
in our approach.  

\end{abstract}

\end{titlepage}

\newpage

\section{Introduction}

The Steiner tree problem is one of the most fundamental and
important problems in Computer Science and Operations Research.
Whereas a $2$-approximation is easily obtained by computing a
minimum spanning tree over the terminals, obtaining algorithms
with an approximation guarantee bounded away from $2$ has proven
to be a non-trivial task.
The problem is known to be inapproximable to within $\frac{96}{95}$,
unless $\mathbf{NP} = \mathbf{P}$~\cite{BP89,CC08}).
There has been a long sequence of
combinatorial approximation algorithms~\cite{GP68,zelikovsky_1993_11_over_6_apx,KZ97,PS00,RZ05},
based on different greedy approaches, culminating in the famous
$1+\frac{\ln(3)}{2}+\epsilon<1.55$ approximation of
Robins and Zelikovsky~\cite{RZ05}.
No further progress was achieved until 
Byrka, Grandoni, Rothvo\ss\ and Sanit\`a~\cite{byrka_2011_steiner} presented 
the first LP-based approach leading to a
$\ln(4)+\epsilon\approx 1.39$ approximation.
A major hindrance in the design of LP-based Steiner
tree algorithms is a rather poor understanding of
potential LP relaxations. In particular, until
the result of~\cite{byrka_2011_steiner}, for no (efficiently
solvable) LP relaxation of the Steiner tree problem was
it known whether the integrality gap was bounded away from $2$.
Intriguingly, even though their $\ln(4)+\epsilon$ approximation
algorithm is based on a particular LP relaxation, its approximation guarantee is
not with respect to the LP solution and does not imply a
$\ln(4)$ integrality gap for the relaxation.
In~\cite{byrka_2011_steiner}, the authors show 
a weaker $\approx 1.55$ integrality gap
using a technique not directly linked to their algorithm.
Chakrabarty et al.~\cite{chakrabarty_2010_integralitygap} provide a simpler alternative proof of the same bound.

The linear relaxation used by Byrka et al., the \emph{directed component-based relaxation}, was introduced 
by Polzin and Vahdati-Daneshmand~\cite{PV03}, based in turn on an equivalent \emph{undirected} component-based
LP introduced by Warme~\cite{warme_1998_hyperspanningtrees}.
It is the undirected version that we will use in this paper.
Another notable relaxation is the partition-based LP introduced by K\"onemann et al.~\cite{koenemann_2011_partition_based_LP}.
In~\cite{chakrabarty_2010_hypergraphic}, Chakrabarty et al.~showed that this relaxation is equivalent to the others mentioned above, and introduced the term ``hypergraphic'' for this family of relaxations.
They also proved that basic solutions are sparse, having support size less than the number of terminals.
The limited understanding of LP relaxations of the
Steiner tree problem is arguably a major barrier in the design
of stronger approximation algorithms.
The goal of this work is to fill this gap by providing a fresh
view on the component-based LP relaxation---one that heavily
exploits methods and results from the theory of matroids
and submodular functions.
More precisely, based on the approach of Byrka
et al.~\cite{byrka_2011_steiner}, we present a deterministic
$\ln(4)+\epsilon$ algorithm
that starts with a solution to the component-based LP relaxation,
iteratively contracts a component and updates the LP
solution.
The algorithm of Byrka et al.~solves the
component-based LP (through a very large extended formulation)
in each iteration after contracting, in order to again obtain
a feasible solution.
By contrast, we show how
the LP can be modified by a simple greedy algorithm over
a well-chosen matroid to achieve the same goal.
This leads to a considerably faster way to update the
LP, but more importantly, we show how the approximation
quality of our approach can be analyzed  with respect
to the initial LP solution.
This implies a bound of the integrality gap of the component-based LP
relaxation of $\ln(4)$.
By comparison, the best known lower bound is $8/7 \approx 1.142$ (e.g., by the example of
 \cite{koenemann_2011_partition_based_LP}). %
Furthermore, we show how the separation problem of the
component-based relaxation can be reduced to computing maximum flows.
Whereas this result is likely to be of independent interest,
it also provides a way to obtain a fast
independence oracle for the matroids
that we use in our approach.
Additionally, we further investigate the special case
of quasi-bipartite graphs, which has played a central
role in the design of approximation algorithms for the
Steiner tree problem, as well as to find $\mathbf{APX}$-hard
problem classes.
Rajagopalan and Vazirani~\cite{RV99} showed that the integrality gap
of the bidirected cut relaxation for such graphs can be bounded by $3/2$.
This was later improved to $4/3$ \cite{CDV08} and to $1.28$~\cite{chakrabarty_2010_integralitygap}.
We obtain a $\frac{73}{60}$ bound for the integrality gap, again matching the approximation factor of~\cite{byrka_2011_steiner}.
Such a bound was previously known only for the case when all edge costs are equal~\cite{chakrabarty_2010_integralitygap}.
Chakrabarty et al.~\cite{chakrabarty_2010_hypergraphic} showed that on quasi-bipartite graphs, the bidirected cut and hypergraphic relaxations are actually equivalent. 
 However their proof is based on a duality argument, and they leave as an open problem the question of converting a solution from the bidirected cut relaxation to the hypergraphic relaxation efficiently (more quickly than simply optimizing the hypergraphic LP).
We present a simple algorithm to perform this transformation; since the bidirected cut relaxation can be solved much more efficiently via a compact extended formulation, this gives a much faster method of solving the hypergraphic LP in the quasi-bipartite case.
Combining this result with the suggested approximation
algorithm, we obtain a significantly faster
$\frac{73}{60}$ approximation than
the one of Byrka et al.~\cite{byrka_2011_steiner},
since we do not need to (repeatedly) optimize the component-based
relaxation by using either the ellipsoid method or a
very large extended formulation.

\section{Discussion of results and techniques}\label{sec:overview}

\subsection{The component-based LP}

Let $G=(V,E)$ be an undirected graph with terminals $R\subseteq V$
and edge costs $c : E \to \mathbb{R}_+$.
A \emph{component} $C$ is simply a subgraph of $G$ with the property that it is a tree spanning $V(C)$, all leaves of $C$ are terminals, and all internal nodes are
non-terminals.
Write $\cost(C) := \sum_{e \in E(C)} \cc{e}$ for the cost of a component $C$.
We will frequently need the terminal set of a component $C \in \comp$, and so by abuse of notation, when we refer to $C$ as a vertex set, we mean the set $V(C)\cap R$ of terminals in $C$. In particular, $\card{C}$ refers to the number of terminals in $C$.

Now let $\comp$ be the set of all components of $G$; we assume that all components contain at least two terminals, else they can be safely removed.
We use the notation $(Z)^+ := \max\{Z, 0\}$.
Then the component-based LP relaxation is as follows~\cite{warme_1998_hyperspanningtrees}:
%
\begin{equation}\tag{\textsc{lp}}\label{eq:kLP}
	\begin{alignedat}{3}
\min \quad&&\sum_{C\in \comp}
         x_C\cost(C)\quad\quad\quad &\\
&&     \sum_{C\in \comp}
         x_C (|S\cap C|-1)^+  ~&\leq~  |S|-1 &\quad\qquad&\forall S\subseteq R,
         S\neq \emptyset \\
&&     \sum_{C\in \comp}
         x_C (|C|-1) ~& =~  |R|-1 \\
&&      x_C ~&\geq~ 0  &&\forall C\in \comp.
\end{alignedat}
\end{equation}
Borchers and Du~\cite{borchers_1997_ksteiner} showed that the optimal \emph{$k$-restricted} Steiner tree, meaning only components with at most $k$ terminals can be used, has cost at most $1+1/\lfloor\log_2 k\rfloor$ times the cost of an optimal Steiner tree. 
Furthermore, when restricting the variables in~\eqref{eq:kLP}
to components with at most $k$ terminals, the resulting
linear program can be solved efficiently, e.g., by solving
a polynomial-size extended formulation~\cite{byrka_2011_steiner}.  
It follows that for any fixed $\epsilon>0$, a ($1+\epsilon$)-approximate
solution to~\eqref{eq:kLP} can be obtained efficiently.
We also point out in Appendix~\ref{appendix:nphardness} that
optimizing $\eqref{eq:kLP}$ exactly is strongly $\mathbf{NP}$-hard (this does not
seem to have been previously observed).

\medskip

The framework of our algorithm is similar to Byrka et al.~\cite{byrka_2010_improved}, and in particular, it is iterative in nature.
They begin by computing a near-optimal fractional solution $x$ to \eqref{eq:kLP}. 
They then sample a component $C$ at random, proportional to its entry $x_C$, and \emph{contract} this component. 
The solution $x$ is no longer feasible to \eqref{eq:kLP} on this new contracted instance, so they re-solve the LP and iterate this procedure until all terminals are connected. 

In their analysis, they show that a single random contraction reduces the cost of the \emph{optimum} Steiner tree by a certain factor in each iteration.
The crucial ingredient is a lower bound on the expected cost of edges that could be removed from an optimum solution after a contraction, while still obtaining a Steiner tree.
To obtain a bound on the integrality gap, we need a stronger result that says even a 
\emph{fractional} solution becomes significantly cheaper after a random contraction.
Even for a fixed set of terminals $Q$, it was unclear how to modify a fractional solution 
in order to preserve feasibility after contraction---a question that had a simple answer in the integral case. \todo{should we say why?}
Our first goal will be to obtain an understanding of the structure of these modifications.

\medskip

While it can be avoided, it significantly simplifies the discussion to consider ``blown up'' versions of solutions to~\eqref{eq:kLP}.
Consider any $x \in \mathbb{Q}^{\comp}_+$, and let $N \in \mathbb{N}$ be such that $x_C \cdot N \in \mathbb{N}$ for all $C \in \comp$. 
The minimal \emph{blowup graph} corresponding to $x$ is the unweighted multigraph defined as follows.
First take the disjoint union of $x_C \cdot N$ disjoint copies of $C$ for each component $C$;
then identify, for each $v \in R$, all the copies of $v$.
The edge costs of $\grphx$ are inherited from $G$ in the obvious way.
See Figure~\ref{fig:blowup} for an example of an LP solution and its associated minimal blowup graph.
Observe then that $\cost(\grphx) = N\cdot\cost(x)$.
Note that $\grphx$ (along with $N$, but this will remain fixed throughout) encodes all the information in $x$. 
In particular, given $\grphx$ we can determine all of its components: these are simply the maximal connected subgraphs that are trees whose leaves are precisely the terminals spanned.
Thus we can define $\compx$ as the set of components of a blowup graph $\grphx$. 
Each component $C \in \compx$ is a subgraph of $\grphx$, but again, we will abuse notation when the context is clear and sometimes use $C$ to refer to just the terminals of $C$.
Thus, e.g., for some $S \subseteq R$, $S \cap C$ refers to the terminals in $C$ that are also in $S$.
We will need slightly more generality in our definition of a blowup graph.
For any $t \in \naturals$, let $G_t$ be the multigraph obtained by first taking $t$ disjoint copies of $G$, and then for each $v \in R$, identifying all copies of $v$.
For a solution $y$ with corresponding minimal blowup graph $\grphy$, we call a multigraph $\grphy'$ a (not necessarily minimal) blowup graph corresponding to $y$ if
\begin{inparaenum}[(i)]
\item $\grphy' \subseteq G_t$ for some $t \in \naturals$,
\item $\grphy' \supseteq \grphy$, and
\item for any distinct terminals $u,v \in R$, there is no $u$-$v$-path in $\grphy'$ that is not already present in $\grphy$.
\end{inparaenum}
Any edges in $\grphy'$ that were not in $\grphy$ we call \emph{pendant} edges.
We will say that a blowup graph $\grphy$ is \emph{feasible} if it corresponds to a feasible solution to \eqref{eq:kLP}; otherwise we call it infeasible.
Note that pendant edges have no effect on feasibility; they will always be removed in what we will later call a ``cleanup'' step.

\subsection{Edge removal after contraction}

Let $\grphx$ be the blowup graph corresponding to some solution $x$.
We are interested in the situation after contracting some full component of $G$.
In order to avoid some annoying technicalities, for now instead of contracting $Q$ we will think of increasing the value of $x_Q$ by $1$. In other words, in terms of the blowup graph, we take $N$ fresh copies of component $Q$ and add it to $\grphx$.
We denote the new blowup graph obtained by $\grphx \contract Q$. 
Formally, $\grphx \contract Q$ is obtained by taking the disjoint union of $\grphx$ and $N$ copies of $Q$, and then identifying all copies of $v$ for each $v \in R$.

It is clear that $\grphx \contract Q$ is not feasible.
We are interested in describing the set of edges $F \subseteq E(\grphx)$ that can be \emph{removed} so that $(\grphx \contract Q) \remove F$ is feasible.

This is the primary reason that it is simpler to work with the blowup graph $\grphx$ rather than $x$; this modification operation is much simpler than an equivalent operation defined on $x$. 
For example, removing a single edge from $\grphx$ can have the effect of splitting up some component $C$ into subcomponents $C_1$ and $C_2$; the corresponding effect on $x$ is to reduce $x_C$ by $1/N$ and increase $x_{C_1}$ and $x_{C_2}$ by the same amount.

Unfortunately, the set of all possible edge removals is not so well behaved.
In order to expose the structure we need, we must consider \emph{minimal} removals. 
Let

	\[ \mathcal{B}_Q = \{ B \subseteq \grd \mid (\grphx \contract Q) \remove B \text{ is feasible, and $B$ is minimal with this property} \}. \]
	Figure~\ref{fig:contraction} shows an example; after a set $B \in \mathcal{B}_Q$ is removed, an edge of the blowup graph becomes pendant, and so can also be removed without affecting feasibility.

One of the most crucial elements of our analysis is the following:
\begin{theorem}\label{thm:isMatroid}
For every component $Q$,
$\mathcal{B}_Q$ forms the set of bases of a matroid $M_Q$.%
\end{theorem}
In particular, it follows that any minimal removal set has the same number
of edges; this number turns out to be $N(|Q|-1)$.
We are able to give a precise description of the matroid $M_Q$ by giving its rank function; 
more details of this will be given in Section~\ref{sec:deeper}.
We can also show that the matroid is a gammoid (a special type of matroid related to flows); see Appendix~\ref{appendix:sep}.
As an aside, we note that $M_Q$ depends only on the terminals of $Q$, and not its structure; we could actually define a matroid $M_S$ for \emph{any} subset $S$ of terminals, but this will not be important for our purposes.

\medskip %

We will now study which edge sets can be removed after
the random contraction of a component. Even though we will finally
present a deterministic algorithm, this analysis will be helpful in guaranteeing the existence of removal sets with
certain properties by an averaging argument.

As before, let $\grphx$ be the blowup graph corresponding to a feasible LP solution $x$.
Upon contracting component $Q$, we may remove some edges in order to again obtain a feasible solution.
In particular, by Theorem~\ref{thm:isMatroid}, we can remove any basis of $M_Q$.
For added flexibility, we allow choosing
a basis $B_Q \in \mathcal{B}_Q$ randomly, according to any distribution we like.
In this case, each edge $e$ will be removed with some probability $q_e$.
The probability vectors that are attainable are simply the convex combinations of incidence vectors of the bases; in other words, 
precisely the vectors in $B(M_Q)$, the base polytope of $M_Q$.

Now consider, as in~\cite{byrka_2010_improved}, randomly contracting a single component, with component $Q \in \compx$ contracted with probability $1/\card{\compx}$.
Note that since each original component $\tilde{Q} \in \comp$ has $Nx_{\tilde{Q}}$ copies in $\compx$, this is the same as contracting a component in $\comp$ with probability proportional to $x_{\tilde{Q}}$.
Again, we allow ourselves to choose an arbitrary distribution over $\mathcal{B}_Q$ for removals on contracting $Q$, and ask what probability vectors $p$ describing edge removal probabilities are attainable.
But any such probability vector is given by some convex combination $\tfrac{1}{\card{\compx}}\sum_{Q \in \compx} q^Q$, where $q^Q \in B(M_Q)$. 
In other words, the attainable probability vectors form precisely the polytope $\remP$ given by the Minkowski sum
\[ \remP = \frac{1}{\card{\compx}}\sum_{Q \in \compx} B(M_Q). \]
This implies that $\remP$ is a polymatroid~\cite{mcdiarmid_1975_rados}; from our knowledge of the rank functions of the $M_Q$'s, we can also describe the rank function of $\remP$, as will be described in detail in Section~\ref{sec:deeper}.

In the following, we use \emph{scaled cost} to refer to costs reduced by a factor of $N$, compensating for the blowup factor.
The goal is to show that the expected scaled cost of removed edges is large, compared to the expected cost of the component that is contracted.
Perfection would be if we could always remove edges of total scaled cost as large as the cost of the contracted component, but of course this is not possible (it would imply an integrality gap of $1$).
Thus we lower our goals slightly.
It \emph{is} possible to show that there is a point $p \in \remP$ with $p_e \geq \tfrac{N}{2\card{\compx}}$ for all $e \in E(\grphx)$.
This gives an expected decrease of $\cost(\grphx)/(2\card{\compx})$ in the LP solution after scaling down, and the expected cost of the contracted component is $\cost(\grphx)/\card{\compx}$; so this implies only an uninteresting bound of $2$ on the integrality gap.
Instead, we must choose the distribution more carefully.

More precisely, we will choose a well-structured subset $\core\subseteq \grd$ and only consider
removal probabilities $p\in \remP$
whose support is contained in $\core$.
The set $\core$ will be chosen to be a minimal subset of $\grd$ whose removal
from $\grd$ disconnects all terminals in the blowup graph.
We call such a set a \emph{splitting set}\footnote{The
complements of splitting sets are sometimes
called \emph{losses}.}.
Interestingly, the family of all splitting sets form the bases
of a cographic matroid, since $K$ is a splitting set precisely when
$E(\grphx)\setminus K$ is a spanning tree in the graph obtained
from $\grphx$ by contracting together all its terminals.
As we will see more formally in the proof of Theorem~\ref{thm:uniform}, when choosing $K$ to be a
splitting set, the set
$\mathcal{B}_Q^K=\{B\in \mathcal{B}_Q\mid B\subseteq K\}$
is nonempty for every $Q$, and so form the bases
of the matroid $M_Q^K$ obtained by restricting $M_Q$ to $K$.
This implies that the polytope 
$\remPcore=\{p\in \remP\mid \supp(p)\subseteq \core\}$
of removal probabilities we consider is nonempty, and
thus forms the base polytope of the polymatroid
obtained by restricting the polymatroid corresponding
to $\remP$ to $K$. 

Once we have chosen some splitting set $K$,
we will call edges in $K$ \emph{\primary\ edges}, and all
other edges \emph{\secondary\ edges}.
To see the reason for this name, recall that the matroid $M_Q$ describes only the \emph{minimal} edge removals upon contracting $Q$. 
However, there may be other removals that are possible; for $B \in \mathcal{B}_Q$, there may be pendant edges in $(\grphx \contract Q) \remove B$ which can be removed without having any effect on feasibility.
Our choice of $\core$ ensures that for any edge $e \in E(\grphx) \setminus \core$, $e$ can be deleted (``cleaned up'') once enough edges of $\core \cap C$ have been removed.
But just as importantly, we can prove
\begin{theorem}\label{thm:uniform}
	If $\core$ is any splitting set, then there is a distribution over $\mathcal{B}^K_{Q}$ for each $Q \in \compx$ such that if $Q$ is chosen uniformly at random from $\compx$, and then $B$ is chosen from $\mathcal{B}_{Q}^K$ according to the chosen distribution, then
	\[ \Pr{e \in B} \geq N/\card{\compx} \qquad \text{for each } e \in \core. \]
\end{theorem}
This is discussed further in Section~\ref{sec:deeper}.
\subsection{The algorithm}\label{sec:alg}
For the accounting in our analysis, we will need to keep track of precisely which edges in $E(C) \cap \core$ must be removed before an edge $e \in E(C) \setminus \core$ can be deleted (cleaned up). 
Define $W(e) \subseteq \core$, the \emph{witness set} of edge $e$, as the unique \emph{minimal} set of edges such that after removing $W(e)$, $e$ becomes a pendant edge and can be cleaned up.
The fact that there exists such a unique set is shown in Lemma~\ref{lem:witness} in the appendix.
We also define $W(e) = \{e\}$ if $e \in \core$.
Figure~\ref{fig:witness} shows an example of a witness set.

We define a \emph{weight} (distinct from the cost) on all \primary\ edges in such a way that the total weight of \primary\ edges equals the total cost of $\grphx$, by charging the cost of a \secondary\ edge to the \primary\ edges in its witness set.
More precisely, let
\[ w(e) = \cc{e}\;\; + \sum_{f \notin \core: e \in \Witness{f}} \frac{\cc{f}}{\card{\Witness{f}}} \qquad \text{for all } e \in \core. \]
The following is an easy consequence of Theorem~\ref{thm:uniform} and the fact that $\sum_{e \in \core}w(e) = \cost(\grphx)$:
\begin{lemma}\label{lem:contrCostToBasis}
Let $K$ be any splitting set.
There exists some component $Q$ such that $\cost(Q) \leq w(B^Q)/N$, where $B^Q$ is a maximum weight basis of $M_Q^K$.
\end{lemma}
For a given $Q$, a maximum weight basis of $M_Q^K$ can be found via a greedy approach;
all that is needed is an independence oracle.
This we can obtain immediately from our understanding of the rank function of $M_Q^K$; it can be computed using submodular function minimization  (see \eqref{eq:rank} in the next section).
However, while polynomial time, this is quite slow.
We can instead exploit the result that $M_Q^K$ is a gammoid, giving a much faster independence oracle based on solving a maximum flow problem; this is discussed in Appendix~\ref{appendix:sep}.

We are now ready to describe precisely our deterministic algorithm, given in Algorithm~\ref{alg}.
In the algorithm, at each stage we choose a component $Q$ and contract it (in the usual sense, yielding an instance with a smaller vertex set).
Thus at intermediate stages of the algorithm, $\grphx$ will be a feasible blowup graph of some contraction of the original graph $G$.
We also emphasize that the witness sets $W(e)$, and hence also the weights $w(e)$, depend on the blowup graph in the particular iteration.
\begin{algorithm}[h]
\SetEndCharOfAlgoLine{.}
\SetKwInOut{Input}{Input}
\Input{Graph $G$ with edge costs $c$ and terminal set $R$, feasible blowup graph $\grphx$, and splitting set $\core$.}
\KwResult{A Steiner tree $T$.}
\BlankLine
$T \gets \emptyset$\;
\While{$T$ is not a Steiner tree}{
Find a component $Q \in \compx$ and maximum weight basis $B \in \mathcal{B}_Q^K$ with $\cost(Q) \leq w(B)/N$\;
\textit{Cleanup:} Let $F = \{ e \notin \core \mid W(e) \subseteq B\}$\;
\textit{Update:} 
$T \gets T \cup Q$,~~~~$\grphx \gets (\grphx \remove B \remove F) / Q$,~~~~$\core \gets \core \setminus B$\;
}
\caption{A deterministic algorithm for Steiner tree demonstrating a $\ln (4)$ integrality gap.}\label{alg}
\end{algorithm}

We now define, for any blowup graph $\grphx$ and splitting set $\core$, a potential function $\pot{\core}{\grphx}$ by
\[ \pot{\core}{\grphx} := \sum_{e \in E(\grphx)} \cc{e}\harmonic{\card{\Witness{e}}}, \]
where $H(\ell) := 1 + 1/2 + \cdots + 1/\ell$ is the harmonic function.
\begin{theorem}\label{thm:potbound}
	For any minimal splitting set $\core$ and feasible blowup graph $\grphx$, Algorithm~\ref{alg} yields a solution of cost at most $\pot{K}{\grphx}/N$.
\end{theorem}
The proof of this theorem (given in Appendix~\ref{appendix:alg}) essentially boils down to showing that in a single step of the algorithm, the expected cost of the contracted component is no larger than the decrease in the potential function scaled down by $1/N$.
Let $\grphx_t$ and $\core_t$ be the blowup graph and splitting set at iteration $t$ of the algorithm, with $B_t$ the selected removal set.
We are able to show that $\pot{\core_t}{\grphx_t} - \pot{\core_{t+1}}{\grphx_{t+1}} \geq w(B_t)$, from which the theorem immediately follows.

From this, we can use an averaging argument to show the $\ln (4)$ integrality gap bound.
Essentially, if $\core$ is chosen randomly from the matroid of possible minimal splitting sets according to an appropriate distribution, it can be shown that 
\[  \E{\pot{K}{\grphx}} \leq \ln (4) \cdot \cost(\grphx). \]
It is also possible to minimize $\Pot$ as a function of $\core$, via a dynamic program. 
The full proof can be found in the appendix: altogether we obtain, recalling $\cost(\grphx) = N\cdot\cost(x)$, 
\begin{theorem}\label{thm:ln4bound}
	For any solution $x$ of \eqref{eq:kLP}, and choosing $\core$ to minimize $\pot{\core}{\grphx}$, Algorithm~\ref{alg} returns a solution of cost at most $\ln(4) \cdot \cost(x)$.
\end{theorem}
We emphasize again that while we have described everything in terms of the blowup graph, it is possible to implement Algorithm~\ref{alg} directly in terms of the LP solution, yielding a polynomial time algorithm. Details will be provided in the full version.

\subsection{Quasi-bipartite graphs}
The situation is much simplified in the case of quasi-bipartite graphs. In this case, we may choose $\core$ to consist of all edges except for the cheapest in each component. This clearly minimizes $\Pot$, and it can be shown that
\begin{lemma}\label{lem:quasi}
	Let $\core = E(\grphx) \setminus E_{\text{\it min}}$, where $E_{\text{\it min}}$ consists of a cheapest edge from every component. Then 
	\[ \Pot \leq \tfrac{73}{60}\cdot\cost(\grphx). \]
\end{lemma}
A $73/60 < 1.217$ bound on the integrality gap immediately follows from Theorem~\ref{thm:potbound}.
One of the major drawbacks of relying on \eqref{eq:kLP}, or any of the hypergraphic LPs, is that solving them is computational intensive; in general, to obtain a $1+\epsilon$ approximation, nothing better than $n^{2^{\Omega(1/\epsilon)}}$ time is known.
This can be improved somewhat to $n^{\Omega(1/\epsilon)}$ in quasi-bipartite graphs, but this is still very slow.
We show how to sidestep this issue and obtain a reasonable running time for quasi-bipartite graphs by instead solving the much more tractable bidirected cut relaxation, which has only $O(n^2)$ variables.
Combined with the fact that we do not need to re-solve the LP in each iteration, we obtain a markedly faster algorithm than the one of Byrka et al.~\cite{byrka_2011_steiner}.

More precisely, we show how a solution to the bidirected cut relaxation can be transformed into a solution to \eqref{eq:kLP} with the same cost, via a natural greedy procedure.
One step of the transformation consists of taking, from a star centered around a Steiner vertex, all arcs with incoming flow and one arc with outgoing flow.
This yields one component for \eqref{eq:kLP}; the capacities are then uniformly reduced on these edges and the process is continued.
The details of this are given in Appendix~\ref{sec:equivalence}.
Previously, \cite{chakrabarty_2010_hypergraphic} showed that the bidirected cut relaxation always has the same objective value as the hypergraphic relaxations, suggesting that such a transformation should exist, but the question remained open. 

\section{Deeper into the matroid structure}\label{sec:deeper}

In this section, we discuss in more detail the heart of our arguments; uncovering the matroid structure of edge removals, and showing appropriate uniform removal probabilities after the random contraction of a component.

In what follows, we will often need to refer to the terminal set of a component $C$, so we will again abuse notation and write, e.g., $|C|$ for the number of terminals in $C$.
Define $\slkx : 2^R \rightarrow \mathbb{N}$ by
\begin{equation}\label{eq:h}
	\slkx(S)  = N(|S|-1) - \sum_{C\in \compx}  (|S\cap C|-1)^+. 
\end{equation}
It is immediate from \eqref{eq:kLP} that $\grphx$ is feasible if and only if
\begin{equation}\label{eq:constraints}
	\slkx(S) \geq 0 \quad \forall S \subseteq R, S \neq \emptyset \quad\quad \text{and} \quad\quad \slkx(R) = 0.
\end{equation}
Indeed, $\slkx(S)$ is, up to scaling, simply the \emph{slack} (or if negative, violation) of the corresponding constraint in \eqref{eq:kLP}.
Two important properties of $\slkx$ are the following:
\begin{lemma}\label{lem:hsubm}
For any blowup graph $\grphx$, 
\begin{compactenum}[i)]
\item\label{item:hsubm}	
	$\slkx$ is intersecting submodular, i.e., for any two sets $S_1, S_2 \subseteq \grd$ with $S_1 \cap S_2 \neq \emptyset$, \[ \slkx(S_1 \cup S_2) + \slkx(S_1 \cap S_2) \leq \slkx(S_1) + \slkx(S_2), \quad \text{and} \]
\item\label{item:boundedInc} for any $F \subseteq E(\grphx)$ and $\emptyset \neq S\subseteq R$,
$\slkx(S)\leq \slk_{\grphx \remove F}(S)\leq \slkx(S)+|F|$.
\end{compactenum}
\end{lemma}
\begin{proof}
\ref{item:hsubm})
	This follows immediately from the fact that for any $C \subseteq R$, the function $S \rightarrow (|S\cap C|-1)^+$ is intersecting supermodular.

	\ref{item:boundedInc})
The removal of any additional edge $e\in E(\grphx)$ from $\grphx$ leads to a split of some component $C$ of $\grphx$ into
subcomponents $C_1,C_2$ with $C_1\cap C_2 = \emptyset$,
$C_1\cup C_2=C$. Hence,
\[ \slk_{\grphx - e}(S)-\slkx(S)=(|S\cap C|-1)^+ - (|S\cap C_1|-1)^+ -(|S\cap C_2|-1)^+ \in \{0,1\}, \]
which leads to $\slkx(S)\leq \slk_{\grphx - e}(S) \leq \slkx(S)+1$. Applying this repeatedly yields the claim.
\end{proof}
An interesting consequence, that essentially follows by intersecting submodularity of $\slkx$
and standard uncrossing techniques, is that any basic feasible solution to
\eqref{eq:kLP} has a support of size bounded by $|R|-1$
(see, e.g., \cite{goemans_2006_minimum} for an example of this reasoning).
For an equivalent version of~\eqref{eq:kLP}, this result was already
obtained through a rather involved technique by
Chakrabarty et al.~\cite{chakrabarty_2010_hypergraphic}.

For convenience, define
\[ \slkF(S) := h_{\grphx \remove F}(S)  = N(|S|-1) - \sum_{C\in \compF} (|S\cap C|-1)^+. \]

The following lemma describes feasibility of $(\grphx \contract Q) \remove F$ in a convenient form, and also shows that we need only consider constraints corresponding to subsets containing $Q$. 
\begin{lemma}\label{lem:hstarfeas}
	The blowup graph $(\grphx \contract Q) \remove F$ is feasible if and only if 
	$\slkF(R) = N(|Q|-1)$ and 
	$\slkF(S) \geq N(|Q|-1)$ for all $S \supseteq Q$.
\end{lemma}
\begin{proof}
	Let $\grphx' = (\grphx \contract Q) \remove F$.
	Then $\grphx'$ is feasible iff $\slk_{\grphx'}(S) \geq 0$ for all $S \subseteq R$, $S \neq \emptyset$, with equality for $S=R$.
	But 
	\begin{equation}\label{eq:blowcompare} 
		\slk_{\grphx'}(S) = \slkF(S) - N(|S \cap Q|-1)^+,
	\end{equation}
 and so this can be equivalently stated as $\slkF(S) \geq N(|S\cap Q|-1)^+$ for all $S \neq \emptyset$, and $\slkF(R) = N(|Q|-1)$.

	All that needs to be proved then is that only the constraints for $S \supseteq Q$ need to be considered.
	So suppose $S$ is a violated set: $\slk_{\grphx'}(S) < 0$.
	Then $S \cap Q \neq \emptyset$, otherwise $\slk_{\grphx'}(S) = \slkF(S) \geq \slkx(S) \geq 0$ by feasibility of $\grphx$.
	But for any such $S$, %
\begin{flalign*}
&& N(|S|-1) - N(|S\cap Q|-1)^+ ~&=~
  N(|S\cup Q|-1) - N(|Q|-1)^+& \\
  \text{and clearly} && \qquad \qquad \qquad
\sum_{C\in \compF} (|S\cap C|-1)^+
~&\leq~ \sum_{C\in \compF} (|(S\cup Q)\cap C|-1)^+.&
\end{flalign*} 
Subtracting and using~\eqref{eq:blowcompare}, we obtain that $\slk_{\grphx'}(S\cup Q) \leq \slk_{\grphx'}(S)$.
Since $S$ was a violating set, so is $S\cup Q$.
\end{proof}

Let $\grd$ be any subset of $E(\grphx)$, and define
$r_Q : \grd \rightarrow \naturals$ by
\begin{equation}\label{eq:rank} r_Q(F) = \min_{S \supseteq Q} \slkF(S). 
\end{equation}
\todo{Maybe some explanation about $r_Q$, maybe a little further down, explaining some intuition why it makes sense that this is the rank function.}
We will show:
\begin{proposition}\label{prop:matroid}
	The function $r_Q$ is the rank function of a matroid of rank $N(|Q|-1)$. %
\end{proposition}
Once we have this, it is straightforward to show that this matroid precisely describes the minimal edge removals:
\begin{theorem}\label{thm:isMatroidAgain}
	The set of bases of the matroid defined by $r_Q$ is precisely $\minrems{Q}$.
\end{theorem}
\begin{proof}
	Let $\bases{Q}$ be the set of bases of the matroid defined by $r_Q$, and 
	consider any $B \in \bases{Q}$. 
	By the definition of $r_Q$, we have that 
	\[ \slkB(S) \geq r_Q(B) = N(|Q|-1) \qquad \text{ for any }S \supseteq Q. \]
	Moreover, by Lemma~\ref{lem:hsubm}~(\ref{item:boundedInc}),
	\[ \slkB(R) \leq \slkx(R) + |B| = N(|Q|-1); \]
	the final equality follows since $|B|=r_Q(B) = N(|Q|-1)$ and $\slkx(R) = 0$ by feasibility of $\grphx$. 
	Thus by Lemma~\ref{lem:hstarfeas}, $(\grphx \contract Q) \remove B$ is feasible.

	Conversely, consider any $B \in \minrems{Q}$. 
	By feasibility and Lemma~\ref{lem:hstarfeas} again, $\slkB(S) \geq N(|Q|-1)$ for all $S \supseteq Q$, with equality for $S = R$.
 Thus $r_Q(B) = N(|Q|-1)$, and so there is some $B' \subseteq B$ with $B' \in \bases{Q}$.
	But then $B'$ is also a feasible removal set by the above, and so by minimality $B'=B$.
\end{proof}

\begin{proofof}{Proof of Proposition~\ref{prop:matroid}}
	First, observe from \eqref{eq:h} applied to the empty blowup graph that 
	\[ r_Q(E(\grphx)) = \min_{S \supseteq Q} N(|S|-1) = N(|Q|-1). \]
We must show that $r_Q$ is increasing, submodular, and satisfies $r_Q(F) \leq |F|$ for all $F \subseteq \grd$.
The fact that $r_Q$ is increasing follows immediately from the definitions of $r_Q$ and $\slkF$; removing a larger set can only increase the slack.
Considering some fixed $S \supseteq Q$, we have by Lemma~\ref{lem:hsubm}~(\ref{item:boundedInc}) that
$\slk_{\grphx - F}(S) \leq \slkx(S) + |F|$.
Thus $r_Q(F) \leq r_Q(\emptyset) + |F| = |F|$ since $r_Q(\emptyset) = 0$ by feasibility of $\grphx$.

Now we come to the main part of the proof, showing that $r_Q$ is submodular.
We must show that for any $F_1 \subseteq F_2 \subseteq \grd$ and $e \notin F_2$, 
	\begin{equation} r_Q(F_1 + e) - r_Q(F_1) \geq r_Q(F_2+e) - r_Q(F_2).
	\end{equation}
	It is clearly sufficient to show this for $F_1$ and $F_2$ differing by a single edge.
	Consider any $S \supseteq Q$ and $i \in \{1,2\}$. 
	The difference 
	\[ \slkless{F_i+e}(S) - \slkles{F_i}(S) = \sum_{C \in \compof{\grphx \remove F_i}} (|C\cap S|-1)^+\;\; - \sum_{C \in \compof{\grphx \remove (F_i+e)}} (|C\cap S|-1)^+ \]
	is one or zero, and it is one precisely if $e$ splits up some component $C \in \compof{\grphx \remove F_i}$ into two components $C_1, C_2$ that both intersect $S$.
	If this is the case for some component in $\compof{\grphx \remove F_2}$, then $e$ will also split up some component in $\compof{\grphx \remove F_1}$ into two pieces both intersecting $S$, since
	$\grphx \remove F_2$ is a subgraph of $\grphx \remove F_1$.
Thus for any $S \supseteq Q$,
\begin{equation}\label{eq:diff}
	\slkless{F_2+e}(S) - \slkles{F_2}(S) \leq \slkless{F_1+e}(S) - \slkles{F_1}(S).
\end{equation}
It also follows that for any $S, S' \subseteq R$ with $Q \subseteq S \subseteq S'$,
\begin{equation}\label{eq:biggerS}
 \slkless{F_1+e}(S) - \slkles{F_1}(S) \leq \slkless{F_1+e}(S') - \slkles{F_1}(S').
\end{equation}

Let $\minimizers{i}$ be the set of terminal subsets containing $Q$ that minimize $\slkles{F_i}(S)$, over all $S \supseteq Q$.
Since $\slkles{F_1}$ is intersecting submodular by Lemma~\ref{lem:hsubm}, there is a unique \emph{maximal} set
$S_1^* \in \minimizers{1}$, meaning $S_1^* \supseteq S$ for all $S \in \minimizers{1}$.
Similarly, there is a unique \emph{minimal} set $S_2^*\in \minimizers{2}$; so $S_2^* \subseteq S$ for all $S \in \minimizers{2}$.
We first show $S_2^*\subseteq S_1^*$.
Notice that for $S\supseteq Q$ with $S\notin \minimizers{1}$ we have
$\slkles{F_2}(S) \geq \slkles{F_1}(S) \geq \slkles{F_1}(S_1^*)+1$,
where the first inequality follows
by Lemma~\ref{lem:hsubm}~\eqref{item:boundedInc}.
Furthermore, $\slkles{F_2}(S_1^*)\leq \slkles{F_1}(S_1^*)+1$,
again by Lemma~\ref{lem:hsubm}~\eqref{item:boundedInc}.
Hence $\slkles{F_2}(S_1^*) \leq \slkles{F_2}(S)\; \forall%
S\supseteq Q, S\notin \minimizers{1}$, and thus $\minimizers{1}$
must contain some minimizers of $\slkles{F_2}$, i.e., 
$\minimizers{1}\cap \minimizers{2}\neq \emptyset$.
Since $S_2^*$ is the minimal set in $\minimizers{2}$ and $S_1^*$
is the maximal set in $\minimizers{1}$ we obtain $S_2^*\subseteq S_1^*$.

We finally have 
\begin{align*}
	r_Q(F_2+e) - r_Q(F_2) &= \slkless{F_2+e}(S^*_2) - \slkles{F_2}(S^*_2)\\
	&\leq \slkless{F_1+e}(S^*_2) - \slkles{F_1}(S^*_2) \qquad \text{by \eqref{eq:diff}}\\
	&\leq \slkless{F_1+e}(S^*_1) - \slkles{F_1}(S^*_1) \qquad \text{by \eqref{eq:biggerS}, since $S_2^* \subseteq S_1^*$}\\
	&= r_Q(F_1+e) - r_Q(F_1). \qedhere
\end{align*}
\end{proofof}
\todo{I think we could discuss the separation here (give the construction), and then say that this leads on to showing that the matroid $M_Q$ is actually a gammoid.}

\subsubsection*{Proof of Theorem~\ref{thm:uniform}}

We first show that for any component $Q$,
\begin{equation*}\label{eq:lastToShow}
r_Q(\core) = N(|Q|-1).
\end{equation*}
This in turn implies that $\mathcal{B}_Q=\{B\in\mathcal{B}_Q \mid B\subseteq K\}$
is nonempty, since the rank of $M_Q$ is
$N(|Q|-1)$ by Proposition~\ref{prop:matroid}, and so $\remP \neq \emptyset$.
Notice that for $S\subseteq R, S\neq \emptyset$, we have $h_K(S)=N(|S|-1)$, since
in $\grphx-K$ all components contain precisely one terminal. Hence,
\begin{align*}
r_Q(K) = \min_{S\supseteq Q}h_K(S) = N(|Q|-1).
\end{align*}

	As already discussed in Section~\ref{sec:overview}, the polytope $\remP$ is simply a weighted Minkowski sum of the base polytopes $B(M_Q)$ for $Q \in \compx$.
	It is well known that the Minkowski sum of matroid polytopes is a polymatroid, and moreover, the rank function of the sum is simply the sum of the rank functions of the summands~\cite{mcdiarmid_1975_rados}. Thus, $\remP$ is the base polytope of a polymatroid with rank function
	\begin{equation}\label{eq:remrank} 
		r = \frac{1}{\card{\compx}}\sum_{Q \in \compx}r_Q. 
	\end{equation}
	To show that the point $p$ given by $p_e = N/\card{\compx}$ for all $e \in \core$ is in $\remP$, we need to show that $r(F) \geq \card{F}\cdot N/\card{\compx}$ for every $F \subseteq \core$.
	Expanding out~\eqref{eq:remrank} and the definition of $r_Q$, and writing $S_Q$ for the subset $S \supseteq Q$ that attains the minimum in~\eqref{eq:rank}, we obtain
	\[ r(F) = \frac{1}{\card{\compx}}\sum_{Q \in \compx} \slkF(S_Q). \]
	We now observe that because $\core$ is a splitting set, $\slkF(R) = \card{F}$. 
	For imagine removing the edges of $F$ from $\grphx$ one by one; $\slk_{\grphx \remove F}(R) - \slkx(R)$ just counts the number of times where a component is split by the deleted edge in this process. But by the nature of minimal splitting sets, this must happen at \emph{every} step---no pendant edges are formed at any stage. Hence $\slkF(R) - \slkx(R) = \card{F}$; moreover, $\slkx(R) = 0$ by feasibility, so indeed $\slkF(R) = \card{F}$.
\todo{Explain this better.}
	Thus to finish the proof, it suffices to show
	\begin{claim}\label{claim:missing}
		$\sum_{Q \in \compx} \slkF(S_Q) \geq N \cdot \slkF(R)$.
	\end{claim}
To prove Claim~\ref{claim:missing}, we replace the function $\slkF$
on the left-hand side of the inequality by a
function $f$ that lower bounds $\slkF$ and is well structured.
More precisely, $f$ is chosen to be a conic combination of a special type
of intersecting submodular functions which we call
\emph{partition functions}:
for any partition
$\mathcal{P}=\{P_1,\dots, P_n\}$ of $R$, the
corresponding partition function $f_{\mathcal{P}}$ is given by 
\begin{equation*}
f_{\mathcal{P}}(S) =
  \left(|\{j\in [n] \mid P_j\cap S \neq \emptyset\}|-1\right)^+%
   \quad \forall S\subseteq R.
\end{equation*}
The following theorem (whose proof can be found in Appendix~\ref{appendix:sub})
guarantees the existence of the function $f$ that we need to
prove Claim~\ref{claim:missing}.

\begin{theorem}\label{thm:submLowerBound}
Let $h:2^U\rightarrow \mathbb{R}_+$ any nonnegative
intersecting submodular function with $h(\{v\}) = 0$ for all $v \in U$. 
Then there is a monotone intersecting submodular
function $f$ of the form
\begin{equation*}
f = \sum_{i=1}^k \lambda_i f_{\mathcal{P}^i},
\end{equation*}
for some $k \in \naturals$, where $\lambda_i> 0$ and $\mathcal{P}^i$ is a partition of $U$ for each $1 \leq i \leq k$, satisfying:
\begin{compactenum}[i)]
\item\label{item:lowerBounda} $f(S) \leq h(S)$ for all $S\subseteq U$, and
\item\label{item:matchRa} $f(U)=h(U)$.
\end{compactenum}
\end{theorem}
Consider the function $\slkF^+$ defined by $\slkF^+(S) = \max\{\slkF(S), 0\}$.
Then $\slkF^+$ differs from $\slkF$ only on the empty set, since $\slkF(S) \geq 0$ for all $S \neq \emptyset$.
Thus $\slkF^+$ is still intersecting submodular, and also nonnegative.
Let $f=\sum_{i=1}^k \lambda_i f_{\mathcal{P}^i}$ be the function
obtained by applying Theorem~\ref{thm:submLowerBound} to $\slkF^+$.
We then have
\begin{equation}\label{eq:beforePartArg}
\begin{aligned}
\sum_{Q\in \compx} \slkF(S_Q)
&\geq \sum_{Q\in \compx} f(S_Q)
\geq \sum_{Q\in \compx} f(Q)\\
&=\sum_{Q\in \Gamma(\mathcal{X})} \sum_{i=1}^k \lambda_i%
   f_{\mathcal{P}_i}(Q)
= \sum_{i=1}^k \lambda_i \sum_{Q\in \Gamma(\mathcal{X})}%
   f_{\mathcal{P}_i}(Q),
\end{aligned}
\end{equation}
where the first inequality holds since $\slkF(S) = \slkF^+(S) \geq f(S)$ for all $S \neq \emptyset$, and
the second inequality holds since $f$ is monotone and $Q\subseteq S_Q$.

As observed by Chakrabarty et al.~\cite{chakrabarty_2010_hypergraphic},
any solution $x$ to~\eqref{eq:kLP} satisfies the following partition
constraints for any partition $\mathcal{P}$ of $R$:
\begin{equation*}
\sum_{C\in \comp} x_C f_{\mathcal{P}}(C) \geq |\mathcal{P}|-1,
\end{equation*}
where $|\mathcal{P}|$ is the number of sets in
partition $\mathcal{P}$. In our blown-up setting this
translates into
\begin{equation*}
\sum_{Q\in \compx} f_{\mathcal{P}}(Q) \geq N(|\mathcal{P}|-1).
\end{equation*}
Combining this observation with~\eqref{eq:beforePartArg} and
using $|\mathcal{P}_i|-1 = f_{\mathcal{P}_i}(R)$,
Claim~\ref{claim:missing} follows since
\begin{align*}
\sum_{Q\in \compx}  \slkF(S_Q) 
&\geq \sum_{i=1}^k \lambda_i \sum_{Q\in \Gamma(\mathcal{X})}%
   f_{\mathcal{P}_i}(Q)
\geq \sum_{i=1}^k \lambda_i N (|\mathcal{P}_i|-1)\\
&= N \cdot \sum_{i=1}^k \lambda_i f_{\mathcal{P}_i}(R)
= N \cdot f(R)
= N \cdot \slkF(R),
\end{align*}
where the last equality follows from property~\eqref{item:matchRa}
of Theorem~\ref{thm:submLowerBound}.

\clearpage

\begin{figure}
\begin{center}
\subfigure[fractional solution $x$]{ \includegraphics[page=1]{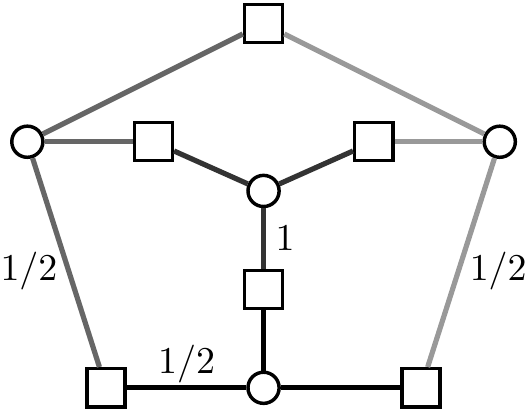}}
\quad \subfigure[blowup graph $\grphx$]{ \includegraphics[page=2]{picture-sources-pics.pdf}} 
\caption{In $(a)$: fractional solution $x$ (components drawn in different gray scales and labelled with their capacity $x_C$). In $(b)$: blowup graph $\grphx$ for $N=2$.  \label{fig:blowup}}
\end{center}
\end{figure}
\begin{figure}
\begin{center}
\subfigure[blowup graph $\grphx \contract Q$]{ \includegraphics[page=3]{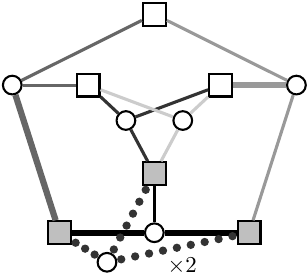}}\quad 
\subfigure[blowup graph $\grphx \contract Q - B$ after removal]{ \includegraphics[page=4]{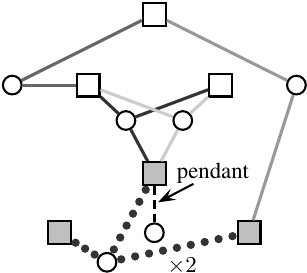}}
\caption{%
In $(a)$: $\grphx \contract Q$, edges in $B \subseteq E(\grphx)$ in bold, copies of $Q$ are dotted, terminals in $Q$ are filled gray. In $(b)$: feasible blowup graph $\grphx \contract Q - B$ (which is not minimal due to the pendant edge that may also be removed). \label{fig:contraction}}
\end{center}
\end{figure}

\begin{figure}
\begin{center}
 \includegraphics[page=10]{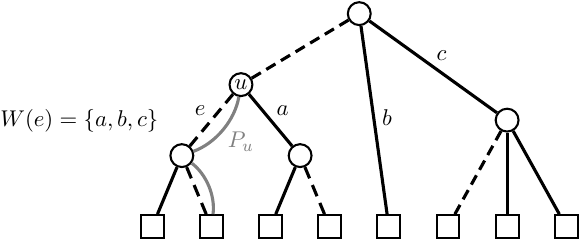}
\caption{Illustration of the definition of $W(e)$. Depicted is some component $C$ with core edges $K$ (solid) and cleanup edges (dashed). 
 \label{fig:witness} }
\end{center}
\end{figure}

\clearpage

\newpage
\appendix
\section{Separation and gammoid structure}\label{appendix:sep}
In this section, we investigate the separation problem for \eqref{eq:kLP}. 
Although it is not necessary, for convenience we will work in the blown up formulation; thus, for a given $\grphx$, our goal is to determine whether \eqref{eq:constraints} is satisfied (see Section~\ref{sec:deeper} for details of this and the definition of $\slkx$). 
In fact, we will do more; for any $Q \subseteq R$, $Q \neq \emptyset$, we will find the most violated set over all $S \supseteq Q$. Given this, we can answer the separation question by checking that \(\min_{S \supseteq \{ v \} } \slkx(S)\) is zero for each choice of $v$ (note that $\slkx(\emptyset) = -N$, and so we must exclude this trivial set from consideration). For each choice of $v$, one max-flow calculation will be required.

The construction is inspired by one for the forest polytope~\cite{picard1982,padberg1983} (see also~\cite[\S51.4]{schrijver_2003_combinatorial}). While what follows is not precisely a generalization (in the case where all components have size $2$, the resulting construction is slightly different), it is similar in spirit.
In the \emph{directed} component-based relaxation, separation via an equivalent flow-based formulation is completely straightforward.
However this does not imply such a formulation for the undirected version.

Since $\slkx$ is an intersecting submodular function, it follows already that $\min_{S \supseteq Q} \slkx(S)$ can be computed in polynomial time~\cite{schrijver_2003_combinatorial}, using submodular function minimization as a black box. 
However, the combinatorial algorithm we demonstrate here, which reduces the separation problem to a max-flow calculation, gives some additional insights (as well as being more efficient).

Let $\grphx$ be the blowup graph of some solution $x$.
First, let 
\[ y_v = |\{C \in \compx: v \in C\}| - N \quad \text{for } v \in R. \]
If $y_v$ is negative for any $v$, it is easily seen that $\grphx$ is not feasible (it corresponds to $x(\delta(v)) < 1$). 
So from now on, we assume $y_v \geq 0$ for all $v \in R$.
Construct a directed multigraph $D=(W,A)$ (we will write $D_\grphx$ if we wish to be explicit on the choice of $\grphx$) as follows. 
Begin with the multigraph $\grphx$, and for each component $C \in \compx$, pick an arbitrary node $r_C$ as the root.
Adjoin a source node $s$ and sink node $t$.
Now orient all edges of $E(C)$ away from $r_C$ for each component, and adjoin the arcs $sr_C$ for each $C \in \compx$, and $vt$ for each $v \in R$.
We assign capacities $z$ to the arcs; $z(vt) = y_v$ for all $v \in R$, and $z(a) = 1$ for all other arcs.

\begin{theorem}\label{thm:separation}
	For any nonempty $Q \subseteq R$, the value of the maximum $s$-$(Q \cup \{t\})$-flow in $D$ is equal to $y(R) + N + \min_{S \supseteq Q} \slkx(S)$.
	More specifically, if $U^*$ is a minimum $(Q \cup \{t\})$-$s$ cut in $D$ with $s \notin U^*$, then $S^* = U^* \cap R$ minimizes $\slkx(S)$ over $S \supseteq Q$, and 
\[  z(\delta^-(U^*)) =  y(R) + N+ \slkx(S^*). \]
\end{theorem}
\begin{proof}
Use $\dirE(C)$ to denote the arcs in $D$ corresponding to $E(C)$ in $\grphx$, and let $A(C) = \dirE(C) \cup \{sr_C\}$.
For any $S \subseteq R$, let $\nu(S) \subseteq W$ be defined by 
\[ \nu(S) = S \cup \{t\} \; \cup \bigcup_{\stackrel{C \in \compx:}{C \cap S \neq \emptyset}} (V(C) \setminus R). \]
\begin{claim}
	For any nonempty $S \subseteq R$, $z(\delta^-(\nu(S))) = \slkx(S) + y(R) + N$. Moreover, for any $U \subset W$ with $s \notin U$, $t \in U$ and $U \cap R = S$, we have $z(\delta^-(U)) \geq \slkx(S) + y(R) + N$.
\end{claim}
\begin{proof}
Consider some component $C \in \compx$. 
If $S \cap C = \emptyset$, then clearly $\delta^-(\nu(S)) \cap A(C) = \emptyset$. 
So suppose $S \cap C \neq \emptyset$. 
If $r_C \in \nu(S)$, then clearly $\delta^-(\nu(S)) \cap \dirE(C) = \emptyset$, and $sr_C \in \delta^-(\nu(S))$. 
On the other hand, if $r_C \notin \nu(S)$ (implying in particular that $r_C$ is a terminal), then $sr_C \notin \delta^-(\nu(S))$ and $|\delta^-(\nu(S)) \cap \dirE(C)| = 1$, since all terminals are leaves of the components they belong to.
In either case, $z(\delta^-(\nu(S)) \cap A(C)) = 1$. 

For any $v \in R$, $sv \in \delta^-(\nu(S))$ if and only if $v \notin S$. Putting this all together, 
\[  z(\delta^-(\nu(S))) =  |\{C \in \compx: C \cap S \neq \emptyset\}| + \sum_{v \notin S} y_v.\]
Now taking \eqref{eq:h} and adding and subtracting $y(S)$, we have
\begin{align*}
\slkx(S) &= N(|S|-1) -\sum_{C \in \compx} (|C \cap S|-1)^+  + \left(\sum_{C \in \compx} |C \cap S| \;-\; N|S|\right) - y(S)\\
&= |\{C \in \compx: C \cap S \neq \emptyset\}|  - N - y(S)\\
&= z(\delta^-(\nu(S)) - N - y(R).
\end{align*}
Now consider any $U$ with $t \in U$, $s \notin U$ and $U \cap R = S$. We again clearly have $sv \in \delta^-(U)$ for all $v \in R \setminus U$, and again $\delta^-(U) \cap A(C) \neq \emptyset$ if $S \cap C \neq \emptyset$. So $z(\delta^-(U)) = \slkx(S) + y(R) + N$.
\end{proof}
By this claim, $z(\delta^-(\nu(S^*)) \leq z(\delta^-(U^*))$; since $U^*$ is a minimum cut, we must have equality.
Then again by the claim, 
\[ z(\delta^-(\nu(S^*))) = \slkx(S^*) + y(R) + N.  \qedhere \] 
\end{proof}

\medskip

We now show how this leads to a description of the matroid $M_Q$ as a gammoid.
Recall the definition of a gammoid: a directed graph $H$ is given, along with two subsets $X,Y \subseteq V(H)$. 
The groundset of the gammoid is $X$, and a set $I \subseteq X$ is independent if there are vertex-disjoint paths from $I$ to some subset of $Y$.
We say in this case that this defines the gammoid \emph{from} $X$ \emph{to} $Y$ in $H$.
It is convenient to observe that by transforming the digraph $H$ appropriately, we can replace vertex-disjoint in the above definition with arc-disjoint, and still characterize gammoids.

We need to slightly tweak the digraph $D$ defined above.
For each $f \in \grd$, there is a corresponding arc $a$ in $D$. Split the arc by adding an additional node $v_f$, producing a ``front'' arc $\afront_f$ with tail $v_f$ and a ``back'' arc $\aback_f$ with head $v_f$. 
We may also remove the node $s$ and all its adjacent arcs.
Call the resulting modified digraph $D'$.

Define the sets
\[ X = \{ v_f \mid f \in \grd \}; \qquad X' =  \bigcup_{C \in \compx} r_C \qquad \text{and}\qquad Y = Q \cup \{t\}. \]
Let $\bgammoid_Q$ be the gammoid defined on $D$ from $X' \cup X$ to $Y$, requiring arc-disjointness rather than vertex-disjointness.
Then define $\gammoid_Q = \bgammoid_Q /X'$; this contraction is also a gammoid.\todo{citation needed?}
By the one-to-one correspondence between $X$ and $E(\grphx)$, we may consider this is a matroid over $\grd$.

\begin{theorem}\label{thm:gammoid}
	For any component $Q$, $\gammoid_Q = M_Q$.
\end{theorem}
\begin{proof}

The rank of a set $U \subseteq X$ in $\gammoid_Q$ is $\rank(U) = \brank(U \cup X') - \brank(X')$, where $\brank$ is the rank function of $\bgammoid_Q$.
Notice that the maximum number of arc-disjoint paths from $X'$ to $Y$ is precisely the max-flow from $s$ to $Q \cup \{t\}$ in $D$.
Thus by Theorem~\ref{thm:separation}, and the definition of $r_Q$, $\brank(X') = r_Q(\emptyset) + y(R) + N$.

Now $\brank(U \cup X')$ is the maximum number of arc-disjoint paths from $U \cup X'$ to $Y$.
But imagine what would happen to $\brank(U \cup X')$ if the arcs $A_U = \{ \aback_f \mid v_f \in U\}$ were removed from $D'$.
Take $P_1, \ldots, P_{\ell}$ to be any maximum collection of arc-disjoint paths from $U \cup X'$ to $Y$ in $D'$.
For some $v_f \in U$,  if some path $P_i$ uses arc $\aback_f$, then certainly no other path emanates from $v_f$, and so we can simply remove the initial segment of $P_i$ before $v_f$ to obtain another maximum collection of disjoint paths that do not use $\aback_f$. 
Repeating this process, we obtain paths $P_1', \ldots, P_{\ell}'$ that do not use any arcs in $A_U$.
But taking $D' - A_U$, and contracting all of $U$ to form the source, yields precisely $D_{\grphx \remove F}$, the digraph for the separation construction corresponding to $\grphx \remove F$.
Thus again by Theorem~\ref{thm:separation}, $\brank(U \cup X') = r_Q(U) + y(R) + N$.
Thus $\rank(U) = r_Q(U)$, and so $\gammoid_Q = M_Q$.
\end{proof}

\section{Proofs for Section~\ref{sec:alg}}\label{appendix:alg}

Let us fix a component $C \in \compx$ and a splitting set $K$. 
By the definition of $K$ (as the complement of a spanning tree in  the graph $\tilde{C}$ obtained by
contracting terminals), every Steiner node $u \in V(C) \setminus C$ has a unique 
path $P_u \subseteq E(C) \setminus K$ of cleanup edges to a terminal that we term $r_u \in C$ (see again Figure~\ref{fig:witness}).
\begin{lemma}\label{lem:witness} 
	For any splitting set $\core \subseteq \grd$, component $C \in \compx$ and edge $e \in E(C) \setminus \core$, let
\[ W(e) = \{ uv \in \core \cap E(C) \mid e \in P_u \}. \]
Then $W(e)$ is the unique minimal subset of $E(C) \cap \core$ whose removal makes $e$ a pendant edge.
\end{lemma}
\begin{proof}
%
Let $\bar{W} \subseteq E(C) \cap K$ be any subset of splitting
edges. If $uv\in W(e)\setminus \bar{W}$  then $e$ remains on a path,
namely $P_u \cup uv \cup P_v \subseteq E(C)\setminus\bar{W}$ between
the terminals, implying that $e$ is not pendant. Thus, any subset
$\bar{W}$ which makes $e$ pendant must contain $W(e)$. 

On the other hand, we claim that $e$ is pendant in $E(C)\setminus
W(e)$. To see this, let $e=ab$ with $e\in P_a$. For $e$ to be pendant,
there would need to be a path $P$ that does not contain $e$ from $a$
to a terminal. But the first edge in $K$ on $P$ must be in $W(e)$,
contradicting the fact that $P\subseteq E(C)\setminus W(e)$.  
\end{proof}

\medskip

\newtheorem*{theorempotbound}{Theorem~\ref{thm:potbound}}

\begin{theorempotbound}
	For any splitting set $K$ and feasible blowup graph $\grphx$, Algorithm~\ref{alg} yields a solution of cost at most $\pot{K}{\grphx}/N$.
\end{theorempotbound}
\begin{proof}
We prove the theorem by showing that the decrease in the potential at any 
iteration is lower bounded by the weight of the edges we remove. More formally, 
consider a given iteration $t$ with current blowup graph $\grphx_t$, 
splitting set $K_t$, and weights $w_t$. let
$Q_t$ be the component to contract and $B_t\in \mathcal{B}_Q^{K_t}$ the edges
to be removed from $\grphx_t$ in this iteration.
At the end of iteration $t$ a new blowup graph
$\grphx_{t+1}$ is obtained with splitting set $K_{t+1}=K_t\setminus B_t$.
We will show
\begin{equation}\label{eq:toShowPot}
\pot{K_t}{\grphx_t}-\pot{K_{t+1}}{\grphx_{t+1}}\geq w_t(B_t).
\end{equation}
This in turn implies the theorem since the potential function at any
iteration, and in particular at the end of the algorithm, is nonnegative.
Therefore, the total weight of all core edges being removed throughout
the algorithm is upper bounded by the potential value of the initial
blowup graph, i.e., $\sum_t w_t(B_t) \leq \pot{K}{\grphx}$.
Furthermore, since at every iteration, $Q_t$ and $B_t$
are chosen such that $\cost(Q_t)\leq w_t(B_t)/N$, we obtain that
the cost of all contracted components---which is the cost of the Steiner
tree our algorithm returns---can be upper bounded by
$\sum_t \cost(Q_t)\leq \frac{1}{N}\sum_t w_t(B_t) \leq \pot{K}{\grphx}/N$,
as desired.
Hence, it remains to prove~\eqref{eq:toShowPot}.

For any edge $e\in \grphx_t$, we denote by $W_t(e)$ its witness
set at the beginning of iteration $t$.
For simplicity, we define $W_{t+1}$ on all of $E(\grphx_t)$, defining
$W_{t+1}(e)=\emptyset$ for $e\in E(\grphx_t)\setminus E(\grphx_{t+1})$.
By definition of the witness sets, we have
\begin{equation}\label{eq:witnesschange}
	W_{t+1}(e)=W_t(e)\setminus B_t \qquad \text{for any } e\in E(\grphx_t). 
\end{equation}
Expanding the left-hand side of~\eqref{eq:toShowPot}, we obtain
\begin{align}
\pot{K_t}{\grphx_t}-\pot{K_{t+1}}{\grphx_{t+1}}
  &= \sum_{e\in E(\grphx_t)} c(e) \big(H(|W_t(e)|)-H(|W_{t+1}(e)|)\big)\notag\\
  &= \sum_{e\in E(\grphx_t)} c(e) \sum_{k=|W_{t+1}(e)|+1}^{|W_{t}(e)|} \frac{1}{k} \notag\\
  &\geq \sum_{e\in E(\grphx_t)} c(e) \cdot \frac{\card{W_t(e)} - \card{W_{t+1}(e)}}{\card{W_t(e)}} \notag\\
  &= \sum_{e\in E(\grphx_t)} c(e) \cdot \frac{|W_t(e)\cap B_t|}{|W_t(e)|}\label{eq:potLHS} \qquad \text{ by~\eqref{eq:witnesschange}}.
\end{align}
Furthermore, by expanding the right-hand side of~\eqref{eq:toShowPot}
using the definition of the weights $w_t$, we obtain
\begin{align}
w_t(B_t) &= \sum_{f\in B_t}\; \sum_{\substack{e\in E(\grphx_t)\\e\in W_t(f)}}%
\frac{c(e)}{|W_t(e)|}\notag\\
&= \sum_{e\in E(\grphx_t)}\; \sum_{f\in B_t\cap W_t(e)}\frac{c(e)}{|W_t(e)|}\notag \\
&= \sum_{e\in E(\grphx_t)}c(e)\frac{|W_t(e)\cap B_t|}{|W_t(e)|}.\label{eq:potRHS}
\end{align}
Inequality~\eqref{eq:toShowPot} finally follows by 
combining~\eqref{eq:potLHS} with~\eqref{eq:potRHS}.
\end{proof}

\medskip

In the following, we show that $K$ can always be chosen s.t. 
$\Pot \leq \ln(4) \cdot \cost(\grphx)$, following the proof of \cite{byrka_2010_improved}. For the sake of a simpler exposition, 
we replace every Steiner node in $\grphx$ of degree higher than 3, 
with a binary tree consisting of cost zero edges in order to obtain nodes that
have degree exactly 3. Suppose we find a suitable pair $(K,F)$ of splitting and cleanup
edges in this auxiliary graph. Then every Steiner node $u$ in the 
original graph has potentially several paths $P_1,\ldots,P_q \subseteq F$ of cleanup edges to terminals. We keep the one path 
minimizing $c(P_i)$ and discard the first edge
of all other paths. This does not increase $\Pot$. Applying this iteratively, we end up with a feasible pair of cleanup edges and splitting edges.

From now on, we assume that every component $C \in \compx$ is a binary tree. 
We pick an arbitrary edge $e_C \in E(C)$ as \emph{root edge}. From any interior
node $u \in V(C) \setminus C$, there are two outgoing edges (these are the edges 
that do not lie on the path from $u$ to the root edge). We randomly pick 
one of these edges as cleanup edge and the other one as splitting edge.
In other words, every interior node $u$ has a unique path of cleanup edges to some terminal and hence, $K$ is a legal splitting set.
Moreover, for every non-root edge $e$ one has $\Pr{e \in K} = \frac{1}{2}$.  
\begin{lemma}
	If $\grd$ is chosen randomly according to the above distribution, 
\[ \E{\Pot} \leq \ln(4) \cdot \cost(\grphx). \]
\end{lemma}
\begin{proof}
Fix a component $C$ and an edge $e \in E(C)$. It suffices to show that
$\E{H( \card{W(e)})} \leq \ln(4)$. 
The root edge is always a splitting edge, thus $|W(e_C)| = 1$. So, let $e$
be a non-root edge and let $v_0,v_1,\ldots,v_{k+1}$ be the path from $e$ to the 
root edge, i.e. $v_0v_1=e$ and $v_kv_{k+1}=e_C$. Let 
\[ X := \max\{ i \mid v_0v_1,v_1v_2,\ldots,v_{i-1}v_{i} \in E(C) \setminus K\} \]
be the number of consecutive cleanup edges on this path, starting 
from $e$ (and $X=0$ if already $v_0v_1 \in \core$). Then $\Pr{X = i} = (\frac{1}{2})^{i+1}$ for $i<k$ and $\Pr{X = k} = (\frac{1}{2})^{k}$.
Furthermore $|W(e)| = X+1$ if $X<k$ and $|W(e)| = k$ otherwise. 
We calculate
\begin{align*}
 \E{H( |W(e)| )} &\leq \sum_{i=0}^{k-1} \Pr{X=i} \cdot H(i+1)+\Pr{X=k} \cdot H(k)\\
 &\leq \sum_{i=0}^{\infty} H(i+1) \cdot \left(\frac{1}{2}\right)^{i+1}\\
 &= \ln(4). 
 \end{align*}

\end{proof}

The above argument can be derandomized by the method of conditional
expectations, and this leads to a proof of
Theorem~\ref{thm:ln4bound}. Another option is to observe that the best
choice of $\core$ can be found in polynomial time, via a dynamic
program as is indicated below.
Combined with the above lemma, this implies Theorem~\ref{thm:ln4bound}.
\begin{lemma}
	A splitting set $K$ minimizing $\Pot$ can be found in polynomial time.
\end{lemma}
\begin{proof}

Since the potential function can be decomposed into terms
corresponding to each component, and a splitting set $K$
consists of the union of splitting sets in each component,
it suffices to consider each component separately.
Hence, let $C$ be any fixed component with vertices $V(C)$
and edges $E(C)$; our goal is to
find a splitting set $K$ for $C$ that minimizes
$\sum_{e\in E(C)} c(e) H(|W(e)|)$.

As usual when applying dynamic programming to problems on trees,
we start by computing tables (to be specified soon) for subtrees 
consisting of a single terminal, and successively combine those tables
until a table for the full tree is obtained, revealing
the optimal splitting set.
To specify the order in which we create tables for larger subtrees
from smaller ones, we direct the edge of the tree $C$ away from an
arbitrarily chosen node in $V(C)$.
We consider the following type of subtrees that we call \emph{partial
trees}. For any vertex $r\in V(C)$ and subset
$U\subseteq \delta^+(r)$ of arcs leaving $r$, the
\emph{partial tree $T_U$ with root $r$} 
is the induced subgraph of $C$ consisting of $r$ and all vertices that 
can be reached from $r$ with paths starting with one of the arcs in $U$.
To simplify notation we also use $T_U$ to refer to the edge set of the partial tree.
Furthermore, let $\overline{T}_U=E(C)\setminus T_U$, and let $R_{T_U}\subseteq R$ denote the terminals contained in the partial tree $T_U$.

To better understand what information should be stored for 
a partial tree $T$, we first briefly discuss how the choice of splitting set $K$ within $T$ impacts the witness sets in $\overline{T}$, 
and vice versa.
We will refer to the choice of core and cleanup edges (i.e., the choice of $K$) within some subset of edges as a \emph{configuration} for that subset.
We distinguish two ways that the root $r$ of $T$ can be connected
to a terminal through cleanup edges: \emph{case (A)} through a path within
the partial tree $T$, and \emph{case (B)} through a path outside of $T$.
Correspondingly, we call a configuration for $T$ a \emph{type (A) configuration} if case (A) holds for the root of $T$, and a \emph{type (B) configuration} otherwise.
Notice that in a type (A) configuration, \emph{every} node within $T$ is connected to a terminal in $R_T$ by cleanup edges.
For a partial tree $T$ we will store two tables, one corresponding
to case (A) and one to case (B).

Consider case (A) and let $P$ be the path of cleanup edges
connecting a terminal in $R_T$ to $r$.
Notice that in this case
$W(e)\subseteq \overline{T}\; \forall e\in \overline{T}$.
Hence, the configuration for $T$ does
not have any impact on the contribution of the edges of $\overline{T}$
to the function $\sum_{e\in E(C)}c(e)H(|W(e)|)$.
However, the witness sets of the edges in $P$ depend on the
configuration for $\overline{T}$, namely every
core edge that can be reached within $\overline{T}$ from $r$ by
following cleanup edges is part of the witness set of any
$e\in P$. Hence, the only information about the configuration for $\overline{T}$
that matters in finding an optimal configuration within $T$
is the number $\alpha$ of core edges in $\overline{T}$ that can be reached
from $r$ through cleanup edges.
Thus for case (A) we want to store a table for $T$ which contains, for
each value of $\alpha\in\{0,\dots, |\overline{T}|\}$, a 
corresponding type (A) configuration that minimizes
$\sum_{e\in T} c_e H(|W(e)|)$.
Here, $|W(e)|$ can be computed
without knowing the precise configuration in $\overline{T}$
(apart from $\alpha$) since
\begin{equation*}
|W(e)| =
\begin{cases}
|W(e)\cap T|         & \text{if } e\in T\setminus P,\\
|W(e)\cap T| +\alpha & \text{if } e\in P.
\end{cases}
\end{equation*}

Now consider case (B) and let $P$ be the path in $\overline{T}$
connecting a terminal to $r$.
In this case the situation is reversed and $W(e)\subseteq T$
for any $e\in T$.
Hence, the configuration for $\overline{T}$
does not have any impact on the contribution of the edges of $T$
to the function $\sum_{e\in E(C)}c(e)H(|W(e)|)$.
However this time, the witness sets of edges on $P$ depends on the
configuration for $T$, namely every
core edges that can be reached within $T$ from $r$ by
following cleanup edges is part of the witness set of any
$e\in P$. Hence, the only information that has to be stored
for $T$ in case (B), in order to describe how the configuration
within $T$ impacts the configuration outside of $T$,
is the number $\beta$ of core edges in $T$ that can be reached
from $r$ through cleanup edges.
Hence, for case (B) we want to store a table for $T$ which contains, for
each value of $\beta\in\{0,\dots, |T|\}$, a corresponding 
type (B) configuration that minimizes
$\sum_{e\in T} c(e) H(|W(e)|)$.

Clearly, if we can compute the (A) table for the full
component $C$, then we are done, since the globally best configuration
is the one minimizing the potential function over all values
of $\alpha$.
Computing type (A) and (B) tables for partial trees corresponding
to single terminals is trivial: table (A) contains one entry corresponding
to $\alpha=0$ of value zero, and table (B) is empty.
There are two constellation we exploit to compute tables for larger partial
trees based on the tables of smaller ones.

The first constellation is the following. Assume that we have tables (A)
and (B) for two partial trees $T_{U_1}$ and $T_{U_2}$ with $U_1\cap U_2=\emptyset$,
and both having root $r$. Then we can compute the two tables
for $T_{U_1\cup U_2}$ from the tables of $T_{U_1}$ and $T_{U_2}$.
This can be done by considering all legal combinations (meaning pairs of configurations that
can be completed to a splitting set)
of one table entry corresponding to $T_{U_1}$ and one corresponding to $T_{U_2}$,
keeping the best ones. Since the size of each table is polynomially bounded
in the input, this can be done efficiently.
We skip the somewhat tedious details for combining those tables which are
based on standard arguments.

In the second constellation, we consider a vertex $r$ and one of its
out-neighbors $v$, i.e., there is an arc directed from $r$ to $v$,
such that both tables for $T_{\delta^+(v)}$ have already been
computed. We can then compute the two tables for $T_{\{rv\}}$ by
considering all legal combinations of an entry of
one of the tables of $T_{\delta^+(v)}$ and the two
possibilities of $rv$ being a core edge or a cleanup edge.

It is easy to observe that starting from the terminals and leveraging
the above two update rules, one can construct both tables for
the full component $C$ efficiently.

\end{proof}

\medskip
For the following Lemma, we assume that the graph $G$ is quasi-bipartite. 
\newtheorem*{lemmaquasi}{Lemma~\ref{lem:quasi}}
\begin{lemmaquasi}
	Let $\core = E(\grphx) \setminus E_{\text{\it min}}$, where $E_{\text{\it min}}$ consists of a cheapest edge from every component. Then 
	\[ \Pot \leq \tfrac{73}{60}\cdot\cost(\grphx). \]
\end{lemmaquasi}
\begin{proof}
Consider a component $C$, which now is a star with edges $e_1,\ldots,e_k$. 
Assume $e_k$ minimizes the cost, then the splitting edges in $C$ are 
$K \cap C = \{ e_1,\ldots,e_{k-1}\}$. First of all, $K$ is
obviously a legal splitting set. Secondly $|W(e_i)| = 1$ for $i\in\{ 1,\ldots,k-1\}$ and $|W(e_k)| = k-1$.
Thus 
\[
\sum_{i=1}^k c(e) \cdot H(|W(e_i)| ) \leq  (k-1 + H(k-1)) \cdot \frac{\cost(C)}{k} \leq \frac{73}{60} \cdot \cost(C),
\] 
using that $1+\frac{H(k-1)-1}{k}$ is maximized for $k=5$. 
The claim follows by summing over all components $C \in \compx$.
\end{proof}

\section{A lower-bound property of nonnegative intersecting submodular functions~\ref{thm:submLowerBound}}\label{appendix:sub}

The main goal of this section is to prove 
Theorem~\ref{thm:submLowerBound}. Before presenting the core
part of the proof we discuss some basic properties of partition
functions, and make some general observations concerning
the statement of Theorem~\ref{thm:submLowerBound} which are
useful to understanding its proof.

Let $U$ be a finite set. We recall that $\mathcal{F}\subseteq 2^U$
is called a \emph{lattice family} if it is closed under unions
and intersections.
A function $\mathcal{F} \rightarrow \R$ is \emph{submodular on $\mathcal{F}$} if $f(A\cup B) + f(A \cap B) \leq f(A) + f(B)$ for all $A,B \in \mathcal{F}$; \emph{supermodular on $\mathcal{F}$}, \emph{intersecting supermodular on $\mathcal{F}$} etc., are defined similarly in the obvious way.

Any partition $\mathcal{P}=\{P_1,\dots, P_n\}$ of $U$ induces
naturally a lattice family $\mathcal{F}_{\mathcal{P}}\subseteq 2^U$
which consists of all possible unions of sets in $\mathcal{P}$.
Consider the coverage function
$\alpha(S)=|\{j\in [n] \mid P_j\cap S\neq \emptyset\}|$, which is clearly submodular.
Notice that we can write $f_{\mathcal{P}}(S) = (\alpha(S)-1)^+$, and in particular 
$f_{\mathcal{P}}(S) = \alpha(S)-1$ for all $S \neq \emptyset$.
Thus $f_{\mathcal{P}}$ is intersecting submodular:
for any $A,B \subseteq U$ with $A \cap B \neq \emptyset$,
\begin{align*}
f_{\mathcal{P}}(A) + f_{\mathcal{P}}(B) &= (\alpha(A)-1) + (\alpha(B)-1)\\
  &\geq \alpha(A\cup B) -1 + \alpha(A\cap B) -1\\
&= f_{\mathcal{P}}(A\cup B) + f_{\mathcal{P}}(A\cap B).
\end{align*}
Furthermore, it is easy to see that 
$f_{\mathcal{P}}$ is intersecting supermodular
on $\mathcal{F}_{\mathcal{P}}$. Hence $f_{\mathcal{P}}$
is \emph{intersecting modular} on $\mathcal{P}$, i.e.,
$f_{\mathcal{P}}(A)+f_{\mathcal{P}}(B)%
=f_{\mathcal{P}}(A\cup B)+f_{\mathcal{P}}(A\cap B)$ for
any intersecting sets $A,B\in \mathcal{F}_{\mathcal{P}}$.

By the above observation, the
function $f$ claimed by Theorem~\ref{thm:submLowerBound}
is by construction intersecting submodular since all $f_{\mathcal{P}^i}$
are intersecting submodular. Similarly, $f$ is monotone
due to the monotonicity of $f_{\mathcal{P}^i}$.
We prove the following slightly stronger version of
Theorem~\ref{thm:submLowerBound}.

\begin{theorem}\label{thm:strongSubmLowerBound}
Let $h:2^U\rightarrow \mathbb{R}_+$ any nonnegative
intersecting submodular function, such that all maximal
sets $S\subseteq U$ with $h(S)=0$ form a partition
$\mathcal{P}^1$ of $U$.
Then there is an intersecting submodular function $f$ of the form
\begin{equation*}
f = \sum_{i=1}^k \lambda_i f_{\mathcal{P}^i},
\end{equation*}
where $k\leq |U|-1$, $\lambda_i> 0 \;\forall i\in [k]$, 
$\mathcal{P}^1,\dots \mathcal{P}^k$ are partitions of $U$ that
become coarser with increasing index, and $f$
satisfies:
\begin{compactenum}[i)]
\item\label{item:lowerBound} $f(S) \leq h(S) \quad \forall S\subseteq U$,
\item\label{item:matchR} $f(U)=h(U)$.
\end{compactenum}
Furthermore, the partitions $\mathcal{P}^i$ together with the coefficients
$\lambda_i$, and hence $f$, can be constructed efficiently.
\end{theorem}

Notice that the condition in Theorem~\ref{thm:strongSubmLowerBound}
stating that the maximal tight sets of $h$ form
a partition of $U$ is equivalent to the property that the family of all
tight sets of $h$ covers $U$, due to the following uncrossing argument.
If the tight sets of $h$ cover $U$ then so
do the maximal tight sets; furthermore, %
for any two intersecting tight sets $A,B\subseteq U$, 
\begin{equation*}
0= h(A)+h(B) \geq h(A\cup B) + h(A\cap B) \geq 0,
\end{equation*}
by submodularity and nonnegativity of $h$; hence $A\cup B$ is also tight.
Hence, this condition is indeed weaker than the one used 
in Theorem~\ref{thm:submLowerBound}, which states that all singletons
must be tight.

\begin{proof}[Proof of Theorem~\ref{thm:strongSubmLowerBound}]
The partitions $\mathcal{P}^1,\dots, \mathcal{P}^k$ and coefficients
$\lambda_1,\dots, \lambda_k$ defining $f$ are obtained as follows.

\begin{center}
\fbox{
\begin{minipage}{0.9\linewidth}
\begin{compactenum}
\item Let $i=1$, $h^1=h$, and $\mathcal{P}^1$ be the maximal
tight sets with respect to $h$.
\item \textbf{While} $h^i(U)>0$:
\begin{compactenum}
\item\label{item:fixlambda} Let $\lambda_i\in \mathbb{R}_+$ be the
maximum value such that
\[ h^i(S)-\lambda_i f_{\mathcal{P}^i}(S) \geq 0
  \quad \forall S\in \mathcal{F}_{\mathcal{P}^i}.\]
\item\label{item:newelems} $h^{i+1} \gets h^i-\lambda_i f_{\mathcal{P}^i}$;
 let $\mathcal{P}^{i+1}\subseteq \mathcal{F}_{\mathcal{P}^i}$ be the
maximal tight sets with respect to $h^{i+1}$.
\item $i \gets i+1$.
\end{compactenum}
\end{compactenum}
\end{minipage}
}
\end{center}

We start by observing that each function $h^i$ encountered
during the algorithm is intersecting submodular over
$\mathcal{F}_{\mathcal{P}^{i-1}}$ (by convention we set 
$\mathcal{P}^0=2^U$), and that
$\mathcal{P}^i$ indeed forms a partition of $U$.
This can easily be verified through an inductive argument.
By assumption $h^1$ is intersecting submodular over $U$,
and $\mathcal{P}^0$ is a partition of $U$.
The intersecting submodularity of
$h^{i+1}=h^i-\lambda_i f_{\mathcal{P}^i}$ over
$\mathcal{F}_{\mathcal{P}^i}$ follows by the intersecting
submodularity of $h^i$ over $\mathcal{F}_{\mathcal{P}^i}$
and the intersecting supermodularity of
$f_{\mathcal{P}^i}$ over $\mathcal{F}_{\mathcal{P}^i}$.
Since $h^{i+1}$ is intersecting submodular
over $\mathcal{F}_{\mathcal{P}^i}$, the maximal tight
sets $\mathcal{P}^{i+1}$ of $h^{i+1}$ in
$\mathcal{F}_{\mathcal{P}^i}$ thus again form a partition of $U$.

The suggested procedure can indeed be
implemented efficiently.
At any iteration $i$ and for any fixed $\lambda>0$,
finding the set $S\in\mathcal{F}_{\mathcal{P}^i}$ minimizing
$h^i(S)-\lambda f_{\mathcal{P}^i}$ is a submodular function
minimization problem. Hence, in step~\eqref{item:fixlambda},
$\lambda_i$ can be found
by using e.g.~binary search, or by applying the parametric
search technique of Megiddo~\cite{megiddo_1983_applying}.

Furthermore, since $f_{\mathcal{P}_i}(S)=0$
for all the sets $S\in \mathcal{F}_{\mathcal{P}^i}$ that
are tight with respect to $h^i$---which are precisely the sets
in $\mathcal{P}^i$---we have $\lambda_i>0$ in each iteration.
By choosing $\lambda_i$ to be maximum in step~\eqref{item:fixlambda},
there is at least one set $S\in \mathcal{F}_{\mathcal{P}^i}$
that is tight with respect to $h^{i+1}$ but not $h^i$. Hence,
$|\mathcal{P}^1|>|\mathcal{P}^2|> \dots$, and the procedures
will terminate.
Let $k$ be the index of the last $\lambda$ that was
set in step~\eqref{item:fixlambda}. Hence, $h^{k+1}(U)=0$,
and $\mathcal{P}^{k+1}=\{U\}$.
Since we start with $|\mathcal{P}^0|\leq |U|$ and the partitions
coarsen at each step, this implies $k\leq |U|-1$.
Additionally, point~\eqref{item:lowerBound} of
Theorem~\ref{thm:submLowerBound} clearly holds by the
termination criterion of the while-loop.

Hence, it remains to prove point~\eqref{item:matchR},
which we prove by showing the following claim through
induction from $j=k+1$ to $j=1$, where $j=1$ corresponds
to the statement~\eqref{item:matchR}:
\begin{equation}\label{eq:indProof}
h^j(S) - \sum_{i=j}^k \lambda_i f_{\mathcal{P}^i}(S) \geq 0
\quad \forall S\in \mathcal{F}_{\mathcal{P}^{j-1}}.
\end{equation}
For $j=k+1$,~\eqref{eq:indProof} clearly holds, since
$h^{k+1}(S)=h^{k}(S)-\lambda_{k} f_{\mathcal{P}^{k}}(S) \geq 0$
$\forall S\in \mathcal{F}_{\mathcal{P}^{k}}$, by
choice of $\lambda_k$. Now let $j\in \{1,\dots, k\}$ and assume
that~\eqref{eq:indProof} holds for all values above $j$.
Let $S\in \mathcal{F}_{\mathcal{P}^{j-1}}$, and we define
$S'\in \mathcal{F}_{\mathcal{P}}^j$ to be the minimal set
in $\mathcal{F}_{\mathcal{P}}^j$ that contains $S$, i.e.,
\begin{equation*}
S' := \bigcup_{\substack{P\in \mathcal{P}^j,\\
P\cap S \neq \emptyset}} P.
\end{equation*}
Notice that
\begin{align}\label{eq:coarsening}
h^{j}(S) =
h^{j}(S) +
\sum_{\substack{P\in \mathcal{P}^j,\\ P\cap S\neq \emptyset}}
h^{j}(P)\geq h^{j}(S')
\end{align}

where the equality holds since all sets in $\mathcal{P}^j$
are tight with respect to $h^j$ by construction, and the
inequality follows by standard uncrossing arguments:
for any set $P\in \mathcal{P}^j, P\cap S\neq \emptyset$,
we have $h^{j}(S)+h^{j}(P)\geq h^i(S\cup P)$
by intersecting submodularity and nonnegativity of $h^i$,
and thus the two terms $h^i(S_{i})$ and
$h^i(P)$ can be replaced by $h^i(S_{i-1}\cup P)$ and
this procedure can be repeated.
In other words, we simply exploit that any nonnegative intersecting
submodular function has the subadditivity property for any
family of sets that are connected when seen as hyperedges
on the given ground set.

The inductive step of the proof of~\eqref{eq:indProof}
finally follows by
\begin{align*}
h^j(S)-\sum_{i=j}^k \lambda_i f_{\mathcal{P}^i}(S)
 &\geq h^{j}(S') - \sum_{i=j}^k \lambda_i f_{\mathcal{P}^i}(S')
 = h^{j+1}(S') - \sum_{i=j+1}^k \lambda_i f_{\mathcal{P}^i}(S')
 \geq 0,
\end{align*}
where the first inequality follows from~\eqref{eq:coarsening}
and the monotonicity of $\sum_{i=1}^k \lambda_i f_{\mathcal{P}^i}(S)$,
and the last one by the inductive hypothesis.
\end{proof}

\section{Equivalence of the hypergraphic and bidirected cut relaxations in quasi-bipartite
  graphs\label{sec:equivalence}}

Let $G=(V,E)$ be a \emph{quasi-bipartite} graph, where
Steiner vertices are not connected by edges (i.e., we have edges only between 
terminals and Steiner vertices or between terminals and terminals). 
 Let $\vec{E}$ be the bidirection of $E$, 
 i.e., for any $\{u,v\} \in E$, $\vec{E}$ contains arcs $(u,v)$ and $(v,u)$.

The \emph{bidirected cut relaxation} with \emph{root} terminal $r \in R$ is  
\begin{equation}\tag{{\textsc{bcr}($r$)}}\label{eq:BCR}
\begin{aligned}
 \min \sum_{e \in \vec{E}} c_e x_e &    \\  
  x(\delta^+(S)) &\displaystyle\geq 1 \quad \forall S \subseteq V \backslash \{ r\}: S \cap R \neq \emptyset \\
  x_e &\geq 0 \quad \forall e \in \vec{E}  
\end{aligned}
\end{equation}

In words: we need to reserve enough capacity in order to support a  %
unit flow from every terminal to the current root $r$.
It was proven in \cite{chakrabarty_2010_hypergraphic} that in 
quasi-bipartite graphs, the value of \eqref{eq:BCR} coincides with
the optimum value of \eqref{eq:kLP}. This was done by lifting an
optimum dual solution for the partition-based relaxation (which is equivalent 
to \eqref{eq:kLP} even in general graphs~\cite{chakrabarty_2010_hypergraphic}) to a dual solution of 
\eqref{eq:BCR}. %
However, the authors of \cite{chakrabarty_2010_hypergraphic} posed as an open question: for a given optimum bidirected cut solution, can a
corresponding primal solution to \eqref{eq:kLP} be directly
extracted without solving \eqref{eq:kLP}? We answer this 
question affirmatively.
 
To avoid an unnecessary case analysis, we split direct edges between terminals 
by inserting a dummy Steiner vertex (we split the edge cost arbitrarily 
among the two parts).
Let $x$ be an optimum solution to \eqref{eq:BCR};
then the \emph{natural decomposition} is as follows. %
For a star with center $u \in V\backslash R$, take an arc $(u,s)$ 
with positive outgoing flow and all arcs $H = \{ (t,u) \mid x(t,u) > 0; \; t\neq s \}$
carrying incoming flow. Let $\epsilon$ be the minimum capacity on 
any of these arcs. Then transfer this capacity into a component $\{ s\} \cup \{ t \mid (t,u) \in H\}$. 
Iterate this process until all capacity has been transferred. 
\todo{This should be expanded slightly.}
The main result of this section is:
\begin{theorem} \label{thm:NaturalDecomposition}
Let $x$ be an optimum solution for \eqref{eq:BCR}. Then the natural decomposition 
yields a feasible optimum solution for \eqref{eq:kLP} with the same objective value.
\end{theorem}

Imagine that we want to ``relocate'' the
root from $r$ to another terminal $r'$. 
We can do this by considering the unit flow from
$r'$ to $r$, and reversing all capacity corresponding to this flow.
This provides a feasible
solution for ${\textsc{bcr}}(r')$ that we term $x^{(r')}$, which again has the same cost (see \cite{GM93} for a proof). 
Note that for any $\{ u,v\} \in E$, the sum $x^{(r)}(u,v) + x^{(r)}(v,u)$ is independent of $r$.
For a Steiner vertex $u \in V \backslash R$, let $N(u) := \{ v \mid \{ u,v\} \in E\}$ be the set of neighbours in the
star with center $u$. %
It suffices to show Theorem~\ref{thm:NaturalDecomposition} for basic solutions, since the
decomposition of a convex combination of capacity vectors equals the 
convex combination of natural decompositions. 
By standard arguments, we may assume that the edge costs are distinct for all edges in the
same star. 

\begin{lemma}\label{lem:BCRincomingflow}
In a star with center $u$ and $r \in N(u)$ one has  $x^{(r)}(r,u) = 0$  and $x^{(r)}(u,s) = 0$ 
for each $s \in N(u)$ with $c(u,s) > c(u,r)$.
\end{lemma}
\begin{proof} 
The flow on arc $(r,u)$ can be removed and the flow on $(u,s)$ arc can be redirected to
$(u,r)$. Both operations would leave the solution feasible and decrease the cost, contradicting optimality. 
\end{proof}
See Figure~\ref{fig:PositiveFlowInBCRStar} for an illustration of the claim.
\begin{figure}[H]
\begin{center}
\includegraphics[page=5]{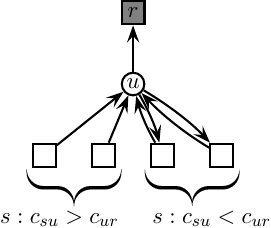}
\caption{Arcs in the optimum solution $x^{(r)}$ that may carry positive flow. \label{fig:PositiveFlowInBCRStar}}
\end{center}
\end{figure}

Next, we consider one iteration of the natural decomposition for a star with center $u$. 
For this reason, insert an extra Steiner vertex $\bar{u}$ into the graph, 
which has an edge $\{\bar{u},s\}$ with $s \in R$ iff there is an edge $\{ u,s\} \in E$ with $x^{(r)}(s,u)>0$.
For $e=(u,s)$, we abbreviate $\bar{e} = (\bar{u},s)$ (see Figure~\ref{fig:TransferingCapacity}).
We define  $c(\bar{e}) := c(e)$ and  $x^{(r)}(e) = 0$ for all $e \in \delta(\bar{u})$.
Note that $x^{(r)}$ is still an
optimum solution.

\begin{lemma} \label{lem:SplittingBCR}
Let $r := \textrm{argmin} \{ c(u,r) \mid r \in N(u) \}$,  $H := \{ (u,r) \} \cup \{ (s,u) \mid s \in N(u) \backslash \{ r\}\}$ and $\varepsilon := \min\{ x^{(r)}(e) \mid e \in H\}$. %
Starting from  $x^{(r)}$, transfer capacity of $\varepsilon$ from all arcs $e \in H$ to $\bar{e}$ and %
term the new capacity reservation $\bar{x}^{(r)}$. 
Then the new capacity vector $\bar{x}^{(r)}$ is a feasible optimum solution 
for \eqref{eq:BCR}.
\end{lemma}
\begin{proof}
We first show that the claim holds for \emph{some} $\varepsilon > 0$ which is small enough.
Consider any cut $S \subseteq V \backslash \{ r\}$ and assume for the sake of a contradiction that  $\bar{x}^{(r)}(\delta^+(S)) < 1$.
For $\varepsilon > 0$ small enough, this may only happen if $S$ was a tight cut before, i.e. $x^{(r)}(\delta^+(S)) = 1$.
Furthermore, any critical cut must contain at least two arcs of the 
form $(s,u)$, i.e., $|N(u) \cap S| \geq 2$. Pick $r' := \textrm{argmin}\{ c(u,r') \mid r' \in N(u) \cap S \}$. 
According to Lemma~\ref{lem:BCRincomingflow}, the flow is $x^{(r)}(e) = 0$ for $e \in (\delta^-(u)\cup\delta^+(u))\backslash H$. %
Since $x^{(r)}(\delta^+(S)) = 1$, the unit flow from $r'$ to $r$ needs all capacities on $(s,u)$
arcs for $s \in N(u) \cap S$. Consequently, when relocating the root to $r'$, the flow
on all these arcs must be turned around completely. In particular $x^{(r')}(u,s)>0$ for $s \in (N(u) \cap S)\backslash\{r'\}$
contradicting Lemma \ref{lem:BCRincomingflow}.
\begin{figure}[H]
\begin{center}
\subfigure[capacities $x^{(r)}$]{ \includegraphics[page=6]{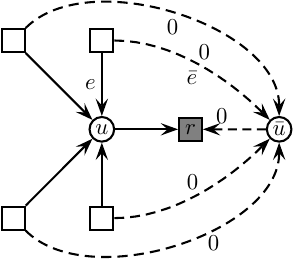}} \hspace{1cm}
\subfigure[capacities $\bar{x}^{(r)}$]{ \includegraphics[page=7]{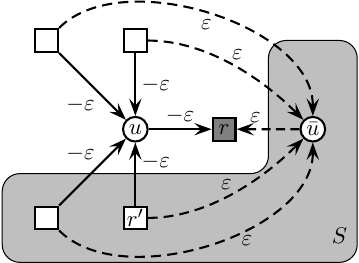}}

\caption{Transferring capacity of $\varepsilon$ according to Lemma~\ref{lem:SplittingBCR}. $(a)$ visualizes capacity in $x^{(r)}$, where newly added edges $\bar{e}$ are dashed. $(b)$ depicts $\bar{x}^{(r)}$ and a potentially critical cut $S$. \label{fig:TransferingCapacity}}
\end{center}
\end{figure}
We conclude that we can choose \emph{some} $\varepsilon > 0$ s.t. $\bar{x}^{(r)}$ is feasible.
But the argument above shows that no cut $S$ can become tight, thus the only limitation on $\varepsilon$
is the arc capacity. The claim then follows. 
\end{proof}
We apply Lemma~\ref{lem:SplittingBCR} iteratively to all stars, adding copies of Steiner nodes as required, until we have a solution $x^*$ (with root $r^* \in R$ chosen arbitrarily) where \begin{inparaenum}[(i)] \item every Steiner node has flow on at most one outgoing arc, and \item the flow on all arcs of a star carrying a nonzero amount of flow is the same. \end{inparaenum}
Then $x^*$ induces a solution to the
\emph{directed component-based relaxation}: %
\begin{eqnarray*}
  \min \sum_{C \in \mathcal{K}, s \in C} \cost(C) \cdot y_{C,s} & &   \\
  \sum_{C \in \mathcal{K},s \in C:  C \cap S \neq \emptyset, s \notin S} y_{C,s} &\geq& 1 \quad \forall \emptyset \subsetneq S\subseteq R\backslash\{r^*\} \\
 y_{C,s} &\geq& 0 \quad \forall C \in \mathcal{K} \; \forall s\in C.  %
\end{eqnarray*}
The solution $y^*$ corresponding to $x^*$ is obtained by setting, for each flow-carrying star with terminal set $C$ and outgoing flow on arc $(u,s)$, $y^*_{C,s} = x^*(u,s)$ (the common flow value in the star).
All other components of $y^*$ are zero.
It is easily checked that $y^*$ is feasible, and has the same cost as $x^*$ (and hence $x$).
Then projecting to the undirected formulation, the vector $(\sum_{s \in C} y_{C,s})_{C \in \mathcal{K}}$ is feasible for \eqref{eq:kLP} (see \cite{PV03}), and moreover corresponds precisely to the natural decomposition described earlier.

\section{$\mathbf{NP}$-hardness for solving the component-based relaxation}\label{appendix:nphardness}

It is well-known that there is a PTAS for solving ~\eqref{eq:kLP}. In other words, for every 
fixed $\varepsilon > 0$, there is a polynomial time algorithm that computes a feasible fractional solution 
to the considered hypergraphic relaxation~\eqref{eq:kLP}, which is within a $1+\varepsilon$ factor
of the optimum fractional value. 
We argue now, that this is best possible (answering the posed question in \cite{chakrabarty_2010_hypergraphic}).
\begin{theorem}
Solving \eqref{eq:kLP} is strongly $\mathbf{NP}$-hard. 
\end{theorem}
\begin{proof}
Let $G=(V,E)$ be a complete graph with terminals $R = \{ s_1,\ldots,s_k\}  \subseteq V$, edge cost 
$c(e) \in \{ 1,2\}$ for all $e \in E$. Bern and Plassmann~\cite{BP89} showed that  
it is $\mathbf{NP}$-hard to decide whether the cost $OPT$ of the 
cheapest Steiner tree is at most a given parameter $Z$. 

We construct another Steiner tree instance $G'=(V',E')$ as follows: 
For each terminal $s_i \in R$ in the original instance, we add 
a terminal $s_i'$ and an edge $s_is_i'$ with cost $c(s_i,s_i') := M$ with $M := 3n^2$ and $n= |V|$.
Furthermore, we downgrade the original terminal to
an ordinary vertex, i.e., we define $R' := \{ s_i' \mid i=1,\ldots,k\}$ as terminal set. %
Let $OPT_f'$ be the value of the optimum fractional solution
of \eqref{eq:kLP} for instance $G'$ (using components of arbitrary size). 

First we show that $OPT \leq Z \Rightarrow OPT_f' \leq Z + k\cdot M$. 
Let $S^*$ be the optimum integral Steiner tree in $G$. 
We add all $s_is_i'$ edges to $S^*$ and consider the emerging tree as 
component with fractional weight $1$ and cost $OPT + k\cdot M$.

Next, we prove that $OPT \geq Z+1 \Rightarrow OPT_f' \geq Z+1 + k\cdot M$ (which in turn implies the claim of the theorem).
Let $x$ be an optimum solution to \eqref{eq:kLP} in $G'$.
For a component $C \in \comp$, we denote $E(C)$ as the contained edges from the 
original graph (i.e. without $s_is_i'$ edges) and by $|C|$ we denote the number of terminals.
Either $C$ contains less than $k$ terminals, or $\cost(E(C)) \geq Z+1$. In any case 
\[
\frac{cost(E(C)) + M}{|C| - 1}
\geq \min\left\{ \frac{Z+1 + M}{k - 1}, \frac{M}{k - 2} \right\} \geq \frac{Z+1 + M}{k - 1}
\]
using that $M = 3n^2$, $k\leq n$ and $Z \leq 2n$.
Now we can bound the cost of the fractional solution as
\begin{eqnarray*}
OPT_f' &=& \sum_{C \in \comp} (\cost(E(C))  + |C|\cdot M)\cdot x_{C} \\
&=& M\sum_{C \in \comp} (|C|-1)x_C + \sum_{C \in \comp} ( \cost(E(C)) + M)x_C \\
&\geq& (k-1)M + \sum_{C \in \comp} \frac{|C| - 1}{k -1}(Z+1+M) x_C
= Z+1+kM 
\end{eqnarray*}
exploiting $\sum_{C \in \comp} x_C ( |C| - 1) = k - 1$. 
\end{proof}
Observe that the above reduction in not approximation preserving. This is not 
surprising, considering the fact that Steiner tree even with edge weights $\{ 1,2\}$
is $\mathbf{APX}$-hard (e.g. by a straightforward reduction from Set Cover with sets of size 3).

\end{document}